\documentclass[a4paper,11pt]{article}
\usepackage[latin1]{inputenc}
\usepackage[T1]{fontenc}
\usepackage{indentfirst}
\usepackage[dvips]{graphicx,color,epsfig}
\usepackage{flafter}
\usepackage{amsmath,amssymb,amsfonts,amsthm}
\newlength{\oldparindent}
\newcommand{\PE}{\mathbb{P}} 
\newcommand{\xrn}{\xrightarrow} 
\newcommand{\nrn}{\rightarrow +\infty} 
\newcommand{\ES}{\mathbb{E}} 
\newcommand{\DNT}{\mathcal{D}^{n}_t} 
\newcommand{\CNT}{\mathcal{C}^{n,i}_{t}}
\newcommand{\CNTS}{\mathcal{C}^{n,i}_{t,s}}
\newcommand{\CN}{\mathcal{C}^{n,i}}
\def\esp#1{\mathbb{E} \left[ {#1} \right]}
\newcommand{\1}{{\bf 1}} 
\def\pr#1#2{\mathbb{P}_{#1} \left[ {#2} \right]} 
\def\VA#1{\left| {#1} \right|}
\def\espcond#1#2{\mathbb{E}_{#1}^n \left[ {#2} \right]}
\newcommand{\Hy}{{\cal H}}
\newcommand{\PRB}{\mathbb{P}} 
\newcommand{\R}{\mathbb{R}}
\newcommand{\bL}{\mathbb{L}} 
\newcommand{\N}{\mathbb{N}}
\newcommand{\F}{{\cal F}} 
\newcommand{\G}{{\cal G}} 
\newcommand{\cLe}{\mathcal{L}^{\epsilon}} 
\newcommand{\cMe}{\mathcal{M}^{\epsilon}} 
\newcommand{\cNe}{\mathcal{N}^{\epsilon}} 
\newcommand{\demi}{\frac{1}{2}} 
\newcommand{\eps}{\varepsilon} 
 
\newtheorem{theorem}{Theorem}

\newtheorem{lemma}[theorem]{Lemma} 
\newtheorem{proposition}[theorem]{Proposition} 
\newtheorem{remark}[theorem]{Remark} 

\author{Alexander Alvarez \footnote{La Habana University, \texttt{alexalvarezh@gmail.com}},
Fabien Panloup\footnote{Institut de Mathématiques de Toulouse et INSA Toulouse,  \texttt{fabien.panloup@.math.univ-toulouse.fr}}, Monique Pontier \footnote{Institut de Mathématiques de Toulouse et Université de Toulouse, \texttt{monique.pontier@.math.univ-toulouse.fr}}, Nicolas Savy \footnote{Institut de Mathématiques de Toulouse et Université de Toulouse,  \texttt{nicolas.savy@.math.univ-toulouse.fr}}}
\title{\textbf{Estimation of the instantaneous volatility}}
\begin{document}

\maketitle

\begin{abstract}
This paper is concerned with the estimation of the volatility process in a stochastic volatility model of the following form: $dX_t=a_tdt+\sigma_tdW_t$, where $X$ denotes the log-price and $\sigma$ is a c\`adl\`ag semi-martingale.  In the spirit of a series of recent works on the estimation of 
the cumulated volatility, we here focus on the instantaneous volatility for which we study estimators built as finite differences of the \textit{power variations} of the log-price. We provide central limit theorems with an optimal rate depending on the local behavior of $\sigma$.
In particular, these theorems yield some confidence intervals for $\sigma_t$. 
\end{abstract}

\\
\noindent \textit{AMS Classification (2000)}: Primary 60F05; Secondary 91B70, 91B82.

\section{Introduction} \label{sec1}
The financial market objects offering a great complexity of modelling,  the development and the study of financial models    has attracted a lot of attention in recent years. For such models, a key parameter is the volatility, which is of paramount importance. The fact that the volatility is not constant has been observed for a long time. Thus, since the famous but too stringent Black and Scholes model, many stochastic volatility models have been introduced. Among them, models where jumps occur are now widely spread in the literature (see $e.g.$ \cite{ConTan} for a review and \cite{CPR08} for a list of recent studies on this topic), mainly because there are able to fit skews and smiles that can not be captured by continuous models.\\
In this paper, we deal with the following kind of model:
$$
X_t =x+\int_0^t a_s ds +\int_0^t  \sigma_s dW_s \quad \forall t\ge0,
$$
where $W$ is a Brownian motion and $\sigma$ is a c\`adl\`ag semi-martingale (assumptions will be made precise in the next section). At this stage, one can remark the main restriction of our model: jumps only occur in the volatility but not in the price. This restriction will be explained in the sequel.\\ 
When such a model is discretely observed, a (now) classical tool for the estimation of the volatility is to make use of the power variations of order $p$ (see next section for details) that have some convergence properties to the cumulated volatility process: $\int_0^t|\sigma_s|^pds$ (when $p=2$, this type of result is only the convergence of the quadratic variations to the angle bracket of the continuous semi-martingale $X$). The study of such estimators of the integrated volatility and its use for the detection of jumps have been deeply studied in the last years (see for instance  \cite{BNSh02,BarShep1,W05} for the continuous setting,  \cite{A04,AJ07,Mancini01,M04,W06,W05,jac1,jac3} for the discontinuous setting and  the more recent papers \cite{JPV08,AlVe10,RR09}).\\ 
Unlike these works, the aim of this paper is to estimate rather the instantaneous volatility. Then, the natural idea is to study estimators which are built as ``derivatives'' of the power variations. More precisely, the proposed estimator of the instantaneous volatility is a normalized relative increment of cumulative volatility estimator, this relative increment being taken on  a smaller and smaller interval. \\
We provide some central limit theorems for the $\sigma_t$-estimator and we exhibit an optimal rate depending on the local behavior of $\sigma$. More precisely, if  a Brownian component 
exists in $\sigma$, the best rate is of order $n^{1/4}$ and otherwise, it depends on the intensity of jumps. In particular, when the jump component has finite-variation, the optimal rate is of order $n^{1/3}$. These central limits lead in particular to some confidence intervals for $\sigma_t$ and to an asymptotic control of the relative error between the estimator and~$\sigma_t$.\\
When jumps occur in the log-price $X$, it seems that we could extend some of the previous announced results by exploiting the fact that convergence properties for the power variations to the cumulated volatility still hold when $p<2$. However, this extension generates some technicalities which are out of ours objectives.\\
The paper is organized as follows. In Section \ref{sec2}, we introduce the model we deal with, we present the different assumptions for this study and we state our main theorems: Central Limit Theorems for the instantaneous volatility.  Section \ref{sec3} is the proof of these theorems. From Theorem \ref{thprincipal}, we easily deduce a confidence interval for the instantaneous volatility, this is shown in  Section \ref{sec4}. Moreover, we stress the fact that %as expected,
 the confidence interval length increases with $p$.  Finally, the volatility estimator is tested on some simulations in Section \ref{sec6}.

\section{Setting and Main Results} \label{sec2}
\noindent We consider a stochastic process $(X_t)_{t\ge0}$ defined on a filtered probability space by: 
\begin{equation} \label{dif}
dX_t=a_tdt+\sigma_tdW_t, \quad t\ge0,
\end{equation}
where $W$ is an $(\mathcal{F}_t)$-adapted Wiener process on $(\Omega, \mathcal{F}, (\mathcal{F}_t), \PRB)$ which satisfies the usual conditions, $a:{\R}_+\rightarrow{\R}$ and $\sigma$ are some c\`adl\`ag $(\mathcal{F}_t)$-adapted processes. Furthermore, $\sigma$ is assumed to be a positive process.\\
Let $T$ be a positive number and assume that $X$ is observed at times $i\Delta_n$ for all $i=0,1,\ldots,[\frac{T}{\Delta_n}]$.
In the sequel, we will assume that $\Delta_n \xrightarrow[n \rightarrow +\infty]{} 0$.\\
In this paper, we want to estimate $\sigma_t$ using the asymptotic properties of the observed discrete increments of $X$. For $p>0$, we denote by $\hat{B}(p,\Delta_n)$, the process of \textit{power variations of order $p$}, $i.e.$ the stochastic process defined by
$$
\widehat{B}(p,\Delta_n)_t :=\sum_{i=1}^{[t/\Delta_n]} \VA{\Delta_i^n X}^p, \quad t \in [0,T]
$$
where $\Delta_i^n X := X_{i\Delta_n} - X_{(i-1)\Delta_n} $.\\

The sequence $(\widehat{B}(p,\Delta_n)_t)_n$ is now classically known as an estimator of $\int_0^t\sigma_s^p ds$. We are going to recall some existing results about the convergence of this sequence but before, we want  to precise the assumptions on $\sigma$ that will be necessary throughout the paper. We introduce the following assumptions depending on parameter $q\in[1,2]$ which is related to the  behavior of the small jumps of $(\sigma_t)$:
\begin{itemize}
\item[$\mathbf{(H^1_q)}$ :] 
$\sigma$ is a positive c\`adl\`ag  semimartingale such that $\sigma_t=|Y_t|$ where $(Y_t)$ satisfies: 
\begin{align*}
dY_s &= b_sds+\eta_1(s)dW_s+\eta_2(s)dW^2_s\\
     &\qquad + \int_{\R}y\1_{|y|\leq 1}(\mu(ds,dy)-\nu(ds,dy)) +\int_{\R}y\1_{|y|>1}\mu(ds,dy),
\end{align*}
where $b$, $\eta_1,\eta_2$ are adapted 
c\`adl\`ag  processes, $\mu$ denotes a random measure on $\R_+\times\R$ with predictable compensator $\nu$ satisfying:
$\nu(dt,dy)=dt F_t(dy)$ and $(\int(1\wedge |y|^q)F_t(dy))_{t\ge0}$ is a locally bounded predictable process.
\end{itemize}
The above assumption on the predictable compensator imply that $(\sigma_t)$ is quasi-left continuous and that the jump component has locally-finite $q$-variation. In order to obtain our main results, we actually need to introduce a little more constraining control of the jump component:
\begin{itemize}
\item[$\mathbf{(H_q^2)}$ :]  For every $ T>0$,
$$\lim_{\varepsilon\rightarrow0}\sup_{t\in[0,T]}\int_{\{|y|\le\varepsilon\}} |y|^q F_t(dy)=0\qquad a.s.$$
\end{itemize}

As an example, if $(Y_t)$ is a solution to the following SDE:
\begin{equation}
\label{exlevy}
dY_t=b(Y_{t^-})dt+\varsigma_1(Y_{t^-})d\tilde{W}_t+\varsigma_2(Y_{t^-})d\tilde{W}^2_t\kappa(Y_{t^-})dZ_t,
\end{equation}
where  $b:\R\rightarrow\R$, $\varsigma:\R\mapsto\mathbb \R$, $\varsigma_2:\R\mapsto\mathbb \R$ and $\kappa:\R\mapsto\mathbb \R$ are some continuous functions with sublinear growth, $(\tilde{W}_{t})_{t\ge0}$ is a Brownian motion and $(Z_{t})_{t\ge0}$ is a centered purely discontinuous L\'evy process independent of $(W_{t})_{t\ge0}$ with  L\'evy measure $\pi$ satisfying $ \int(|y|^q\wedge 1) \pi(dy)<\infty,~q\in[1,2],$ 
then Assumptions  $\mathbf{(H^1_q)}$ and $\mathbf{(H^2_q)}$ hold.

\noindent Before going further, we also need to remind the definition of stable convergence that we denote by ${\cal L}-s$. We say that a sequence of random variables $(Y_n)$ converges $stably$ to $Y$ or $Y_n\overset{{\cal L}-s}{\Rightarrow} Y$, if there exists an extension $(\tilde{\Omega},{\tilde{\cal F}},\tilde{\PE})$ of $(\Omega,{\cal F},\PE)$ and a random variable $Y$ defined on $(\tilde{\Omega},{\tilde{\cal F}},\tilde{\PE})$ such that for every bounded measurable random variable $H$, for every bounded continuous function $f$,
$\esp{H f(Y_n)} \rightarrow \tilde{\ES}[H f(Y)]$ when $n\rightarrow+\infty$ where $\tilde{\ES}$ denotes the expectation on the extension.\\

\noindent Now, we can recall two results (adapted to our context) about the asymptotic properties of $(\widehat{B}(p,\Delta_n)_t)_n$, from L\'epingle \cite{lepingle} (see also \cite{jac3}, Theorem 2.4) and A\"it Sahalia and Jacod \cite[Theorem 2]{jac1} respectively. On the same topic, we can also quote \cite{BGJPS,BarShep1}.

\begin{proposition} \label{lepi}
Assume $\mathbf{(H_2^2)}$. Let $p$ be a positive number and set $m_p:=\esp{|U|^p}$ where $U\sim{\cal N}(0,1).$ Then, 
locally uniformly in $t$,
$$\Delta_n^{1-\frac{p}{2}}\hat{B}(p,\Delta_n)_t \xrightarrow[n \rightarrow +\infty]{\PRB} m_p A(p)_t\quad\textit{with}\quad A(p)_t=\int_0^t\sigma_s^pds.$$
\end{proposition}

\begin{proposition} 
\label{aitjactheo}
Let $p\ge2$ and assume  Assumption 1 of \cite{jac1}. 
 Then, the sequence of continuous processes $(Y(n,p))_{n \in \N}$ defined for any $n \in \N$ by
$$
Y(n,p)_t:=\frac{1}{\sqrt{\Delta_n}}\Big(\Delta_n^{1-\frac{p}{2}}\hat{B}(p,\Delta_n)_t-m_p A(p)_t\Big), \quad t \geq 0,
$$
converges stably to a random variable $Y(p)$ on an extension $(\tilde{\Omega},\tilde{\cal F},
(\tilde{\cal F}_t),\tilde{\PRB})$ of the original filtered space $({\Omega},{\cal F},
({\cal F}_t),{\PRB})$ such that, for any $t \geq 0$, conditionally on ${\cal F}$, $Y(p)_t$ is a centered Gaussian variable 
with variance $\tilde{\mathbb{E}}[Y(p)_t^2 ~|~ {\cal F}]=(m_{2p}-m_p^2)A(2p)_t$.
\end{proposition}
\noindent Looking at these results, it is natural to try to estimate $\sigma_{t^{}}^p$ by  the following statistic: $(\Sigma(p,\Delta_n,h_n)_t)$ defined for every $t \leq \bar{T}$ with $\bar{T} = T-h_1$ by:
\begin{equation}
\label{defSigma}
\Sigma(p,\Delta_n,h_n)_t :=\frac{\Delta_n^{1-\frac{p}{2}}\big(\hat{B}(p,\Delta_n)_{t+h_n}-\hat{B}(p,\Delta_n)_t\big)}{m_p h_n}.
\end{equation}
\noindent Actually, this estimator is the mean of $p$-variations in a window of length $h_n$ where $(h_n)$ is assumed to be a non-increasing sequence of positive numbers such that $h_n$ tends to 0.
%$h_n \xrightarrow[n \rightarrow +\infty]{} 0$.\\

\noindent We are now able to state our main results.

\begin{theorem}
 \label{thprincipal} 
Let $p=2$ or $p\ge3$ and let $(X_t)$ be a stochastic process solution to \eqref{dif}. Assume  $\mathbf{(H^1_2)}$ 
and $\mathbf{(H_2^2)}$. Assume that $\Delta_n=o(h_n)$. Then,\\
\noindent (i)
If $h_n/\sqrt{\Delta_n}\rightarrow0$, $\forall t\in[0,\bar{T}]$, 
\begin{equation}
\label{premiereconvergence}
\sqrt{\frac{h_n}{\Delta_n}}
(\Sigma(p,\Delta_n,h_n)_t-\sigma_t^p)
\xrightarrow[n\rightarrow+\infty]{{\cal L}-s}
\sqrt{\varphi_1(p,t,\sigma)} ~~U,
\end{equation}
where, conditionally on ${\cal F}$,  $U$ is a standard Gaussian random variable  and $\varphi_1(p,t,\sigma)=\frac{m_{2p}-m_p^2}{m_p^2}\sigma_t^{2p}$.\\
\noindent (ii) If $\sqrt{\Delta_n}/h_n\rightarrow\beta\in\R_+$, $\forall t\in[0,\bar{T}]$,
\begin{equation}
\label{convergebis}
\frac{1}{\sqrt{h_n}}
(\Sigma(p,\Delta_n,h_n)_t-\sigma_t^p)\xrightarrow[n\rightarrow+\infty]{{\cal L}-s}
\sqrt{\beta^2\varphi_1(p,t,\sigma)+\varphi_2(p,t)} ~~U,
\end{equation}
where $\varphi_1(p,t,\sigma)$ and $U$ are defined as before and,\\
$\varphi_2(p,t)=
\frac{p^2}{3} (\sigma_t)^{2p-2}\|\eta\|^2(t))$ with $\|\eta\|^2(t) = \eta_1^2(t)+\eta_2^2(t)$.\\
\end{theorem}
\noindent Note that when the drift term $a$ is null,  the result is valid even if
 $2<p<3.$  Otherwise, the drift contributes in a bias for the estimator that is not
negligible in case  $2<p<3.$\\
In the second result, we assume that there is no Brownian component in the volatility, $i.e.$ that $\eta_1=\eta_2=0$ and that the jump component has locally $q$-finite variation. In this case, we show that we can alleviate the constraint on the sequence $(h_n)$ (see also Remark \ref{rk15}).
\begin{theorem}
 \label{thprincipal2} 
Let $p=2$ or $p\ge3$. and let $(X_t)$ be a stochastic process solution to \eqref{dif}. 
Assume  $\mathbf{(H^1_q)}$
and $\mathbf{(H^2_q)}$ with $q\in[1,2]$ and suppose that $\eta_1=\eta_2=0$. Assume that $\Delta_n=o(h_n)$. Then,\\ 
\noindent (i) If $q\in(1,2]$, if $\limsup_{n\rightarrow+\infty}h_n^{1/2+1/q}/\sqrt{\Delta_n}<+\infty$, $\forall t\in[0,\bar{T}]$, 
\begin{equation}\label{tioru}
\sqrt{\frac{h_n}{\Delta_n}}
(\Sigma(p,\Delta_n,h_n)_t-\sigma_t^p)\xrightarrow[n\rightarrow+\infty]{{\cal L}-s}
\sqrt{\varphi_1(p,t,\sigma)} U,
\end{equation}
where $\varphi_1(p,t,\sigma)$ and $U$ are defined as in Theorem \ref{thprincipal}.\\
\noindent (ii) Assume that $q=1$. If $\lim_{n\rightarrow+\infty}h_n^{3}/\Delta_n=0$, \eqref{tioru} holds.\\
If $\lim_{n\rightarrow+\infty}h_n^{3}/\Delta_n=\beta\in\R_+^*$ and if $(\int_{\{0<|y|\le1\}}y F_t(dy))_{t\ge0}$ is càglàd (left-continuous with right limits), then, $\forall t\in[0,\bar{T}]$,
%\begin{equation}\label{thetato}
%\theta_t^0:= p\sigma_t^{p-1}\left(b_t-\lim_{s\searrow t}\int_{\{0<|y|\le1\}}y F_s(dy)\right)+\frac{p(p-1)}{2}\sigma_t^{p-2}\|\eta\|^2(t),
%\end{equation}
\begin{equation}
\label{partconv}
\sqrt{\frac{h_n}{\Delta_n}}
(\Sigma(p,\Delta_n,h_n)_t-\sigma_t^p)\xrightarrow[n\rightarrow+\infty]{{\cal L}-s}
\sqrt{\varphi_1(p,t,\sigma)} U +\frac{\beta}{2}p\sigma_t^{p-1}\left(b_t-\lim_{s\searrow t}\int_{\{0<|y|\le1\}}y F_s(dy)\right).
\end{equation}
\end{theorem}
\noindent Following Tauchen and Todorov \cite{TT08} and according to concrete data, pure jump volatility process could be a more convenient model. In such a case, it seems that the right theorem to be applied is Theorem~\ref{thprincipal2}. \\
\noindent In cases (i) in both Theorems, we get for all $t>0$,
\begin{equation}
\label{CTL2}
\sqrt{\frac{h_n}{\Delta_n}}\left(\frac{\Sigma(p,\Delta_n,h_n)_t}{\sigma_{t}^p}-1\right)
 \xrightarrow[n \rightarrow +\infty]{{\cal L}-s} \frac{\sqrt{m_{2p}-m_p^2}}{m_p} U,
\end{equation}
where $U\sim {\cal N}(0,1)$ and $U$ is independent of ${\cal F}_{t}$.
This  result is  enough to obtain an estimation of $\sigma_t^p$ and to obtain a confidence  interval for it, together the convergence rate.

\begin{remark}
\label{rk15}
It must be stressed here that the convergence rate depends on the balance between the frequency of observations and the length $h_n$ of the window. Hence, the following considerations  justify the choice of a ``good pair"  $(h_n,\Delta_n)$:
let $p=2$ or $p\geq 3$ and assume   $\Delta_n=o(h_n)$.
Considering the window width $h_n,$
 $r_n:=\frac{h_n}{\Delta_n}$  corresponds to the number of observations on the interval $[t,t+h_n]$. Suppose $\Delta_n=\frac{1}{n}$ and  $r_n:=n^\rho,~0<\rho<1,$ then $h_n=n^{\rho-1}.$
In this scheme, assuming $\mathbf{(H^1_2)},$ 
$\mathbf{(H_2^2)}$,  Theorem \ref{thprincipal} yields the following convergence rates: 
\begin{itemize}
	\item[(i)] $\rho<\demi$ yields a convergence rate of order $n^{\rho/2},$
  \item[(ii)] $\rho\ge\demi$ yields a convergence rate of order  $n^{(1-\rho)/2}.$
\end{itemize}
\noindent 
In case $\eta_1=\eta_2=0,$ under Hypotheses $\mathbf{(H^1_q)}$ and
$\mathbf{(H_q^2)}$ with $~1\leq q\leq 2,$ Theorem \ref{thprincipal2} yields the following convergence rates: 
\begin{itemize}
	\item[(i)] if $1<q\leq2,$  $\rho\leq\frac{2}{2+q},$ yields a convergence rate of order  $n^{\rho/2},$
\item[(ii)] the same convergence rate occurs in case $q=1$, $~\rho\leq \frac{2}{3};$
the best convergence rate is of order $n^{1/3}$, obtained for $\rho=2/3$. 
As an example in such a case, let us choose $r_n=n^{2/3}\sim 300.$ It means $300$ data which can be the daily observations  and globally $n = 300^{3/2}\sim 5200.$ 
This may correspond to a realistic data set.
\end{itemize}
%{\color{blue}Finally, let us remark that speed rate could be improved using some weighting factors
%in the sum $\Sigma(p,\Delta_n,h_n)_t$ giving greater weight to observations near  $t,$ for instance think of Robert and %Rosenbaum~\cite{RR09}.}
\end{remark}

\section{Proofs} \label{sec3}
In every proofs $C$ or $C_p$ are  constants which can change from a line to another.
In order to make the notations easier to handle, we will denote by:
\begin{align*}
\DNT &= \left\{ i \in \N,~\left[ \frac{t}{\Delta_n} \right]+1 \le i \le \left[ \frac{t+h_n}{\Delta_n} \right] \right\},&
\CN  &= \left\{ u \in \R,~ (i-1)\Delta_n \le u \le i\Delta_n \right\},\\
\CNT  &= \left\{ u \in \R,~(i-1)\Delta_n\vee t \le u \le i\Delta_n \right\},&
\CNTS  &= \left\{ u \in \R,~(i-1)\Delta_n\vee t \le u \le s \right\}.
\end{align*}

\subsection{Decomposition of the error}
\label{sectiondecomp}
\noindent Following \cite{jac3}, we first decompose $\Sigma(p,\Delta_n,h_n)_t-\sigma_t^p$~as follows:
\begin{equation} 
\label{decopre}
\Sigma(p,\Delta_n,h_n)_t-\sigma_t^p=\frac{Z^{(n,p)}_{t + h_n} - Z^{(n,p)}_t}{m_p h_n}+\Big(\frac{1}{r_n}\sum_{\DNT} \sigma^p_{i\Delta_n}-\sigma_{t}^p\Big),
\end{equation}
where $r_n = h_n / \Delta_n$ and 
$$
Z^{(n,p)}_{t} := \Delta_n^{1-\frac{p}{2}} \hat{B}(p,\Delta_n)_t - m_p  \sum^{[t/\Delta_n]}_{i=1} \Delta_n\sigma^p_{i\Delta_n}.
$$
On the one hand, denoting by $\espcond{i-1}{*}$ the conditional expectation with respect to $\F_{(i-1) \Delta_n}$, one can notice that  $
\espcond{i-1}{\VA{\sigma_{(i-1) \Delta_n} \frac{\Delta^n_i W}{\sqrt{\Delta_n}}}^p } = \sigma_{(i-1) \Delta_n}^p m_p,$
and it is easy to checks that
$$
\frac{Z^{(n,p)}_{t + h_n} - Z^{(n,p)}_t}{h_n}=\Lambda_1^n(t)+\Lambda_2^n(t)+\Lambda_3^n(t),
$$
with 
\begin{align*}
\Lambda_1^n(t)&:=\frac{\Delta_n}{h_n} \sum_{i\in\DNT} \left(\VA{\frac{\Delta_i^n X}{\sqrt{\Delta_n}}}^p- \VA{\sigma_{(i-1) \Delta_n} \frac{\Delta^n_i W}{\sqrt{\Delta_n}}}^p - \espcond{i-1}{\VA{\frac{\Delta_i^n X}{\sqrt{\Delta_n}}}^p - 
\VA{\sigma_{(i-1) \Delta_n} \frac{\Delta^n_i W}{\sqrt{\Delta_n}}}^p }\right),\\
\Lambda_2^n(t)&:=
\frac{\Delta_n}{h_n} \sum_{i\in\DNT} \sigma^p_{(i-1) \Delta_n} 
\left( \left| \frac{\Delta^n_i W}{\sqrt{\Delta_n}} \right|^p-m_p \right) ,\\
\Lambda_3^n(t)&:= \frac{\Delta_n}{h_n}  \sum_{i\in\DNT} \left(\espcond{i-1}{\VA{\frac{\Delta_i^n X}{\sqrt{\Delta_n}}}^p} - m_p  \sigma^p_{(i-1)\Delta_n}\right)
 +\frac{\Delta_n}{h_n}m_p \left(\sigma^p_{[(t+h_n)/\Delta_n]}-\sigma^p_{[t/\Delta_n]}\right).
\end{align*}
On the other hand, let us now decompose the second part of \eqref{decopre}. It\^o's formula applied to $x \to |x|^p$ with $p\ge1$ yields for every $i \ge [t/\Delta_n]+1$:
\begin{equation*}
|Y_{i\Delta_n}|^p = |Y_{t}|^p+ A_{i\Delta_n}-A_t + M_{i\Delta_n}-M_t, \qquad\text{with},
\end{equation*}
\begin{align*}
M_t 
&= \int_0^t p~{\rm sgn}(Y_s)|Y_{s}|^{p-1}\eta_1(s)dW_s+\int_0^tp~{\rm sgn}(Y_s)|Y_{s}|^{p-1}\eta_2(s)dW_s^2,\\
&A_t= \int_0^t \theta_s ds+ \int_0^t \int_{|y| \leq 1} p~{\rm sgn}(Y_s)|Y_{s^-}|^{p-1} (\mu-\nu)(ds,dy)\\
&\quad+ \sum_{0< s \le t} \left( |Y_{s-} + \Delta Y_s |^p - |Y_{s-}|^p -p~ {\rm sgn}(Y_s)|Y_{s-}|^{p-1} \Delta Y_s \1_{|\Delta Y_s| \leq 1} \right),
\end{align*}
and 
$$\theta_s= p~{\rm sgn}(Y_s)|Y_{s}|^{p-1}b_s+\frac{p(p-1)}{2}|Y_{s}|^{p-2}\|\eta\|^2(s).$$
Then, it follows that
\begin{equation*}
\frac{1}{r_n}\sum_{i\in\DNT} \sigma^p_{i\Delta_n}-\sigma_{t}^p=\Lambda_4^n(t)+\Lambda_5^n(t), \qquad\text{with},
\end{equation*}
\begin{align*}
&\Lambda_4^n(t)=\frac{1}{r_n} \sum_{i\in\DNT}([(t+h_n)/\Delta_n]-i+1) (M_{i\Delta_n}-M_{(i-1)\Delta_n\vee t})\\
&\Lambda_5^n(t)=\frac{1}{r_n} \sum_{i\in\DNT}([(t+h_n)/\Delta_n]-i+1)(A_{i\Delta_n}-A_{(i-1)\Delta_n\vee t}). 
\end{align*}
\subsection{Preliminary Lemmas}
In this section, we establish a series of useful lemmas for the sequel of the proof. First, we  show in Lemma \ref{appoxM}  that it is enough to prove the main results under $\mathbf{(H^2_q)}$ and the following assumption:
\begin{itemize}
	\item[${\mathbf{(SH)_q}}$] $a$,$b$, $\eta_1$, $\eta_2$, and $\int_0^.\int(|y|^q\wedge1) F_s(dy) ds$ are bounded  and there exists $M>0$ such that $F_s([-M,M]^c)=0$ $a.s.$ $\forall s\ge0$.
\end{itemize}

\begin{lemma}
\label{appoxM}
Assume that the conclusions of Theorem \ref{thprincipal} and \ref{thprincipal2} hold for every $(X,\sigma)$ satisfying  $\mathbf{(SH)_q}$ and $\mathbf{(H^2_q)}$ (with $q\in[1,2]$ depending on the statement). Then, the conclusions hold for every $(X,\sigma)$ satisfying  $\mathbf{(H^1_q)}$ and $\mathbf{(H^2_q)}$ with $q\in[1,2]$.
\end{lemma}
\noindent The proof of this lemma is based on a classical localization procedure and is done in the Appendix (see Section \ref{preuveannexe}).\\
As a consequence of the preceding lemma, we now work under $\mathbf{(SH)_q}$. In the following preliminary result, we state a series of useful properties on $\sigma$ under this assumption.
\begin{lemma} \label{lemmecontrol}
Assume $\mathbf{(SH)_2}$. 
\begin{itemize}
	\item[(i)] For every $T>0$ and every $r>0,$
	  \begin{equation} \label{EspSup}
    \esp{\sup_{0\leq t\leq T}(\sigma_t)^r} <\infty.
    \end{equation} 
  \item[(ii)] For every $0 \le s \le t \le  T$  such that $|t-s|\leq 1$, it exists a deterministic constant $C_T > 0$ such that:
    \begin{equation}  \label{difsigmap}
    \esp{|\sigma_t-\sigma_{s}|^r~|~\F_s } \leq C_T|t-s|^{1\wedge \frac{r}{2}},~\quad \forall r>0. 
    \end{equation}
    \begin{equation}  \label{majintsigmaW}
    \esp{ \left|\int_s^t\sigma_udW_u\right| ^q} \leq  C_T|t-s|^{\frac{q}{2}},\quad\forall  q>0.
    \end{equation}
    \begin{equation} \label{majintsigmaW2}
    \esp{ \left|\int_s^t(\sigma_u-\sigma_s)dW_u\right| ^q} \leq  C_T|t-s|^{q\wedge (\frac{q}{2}+1)},\quad\forall  q>0.
\end{equation}
\end{itemize}
\end{lemma}
\noindent The proof of this lemma is based on standard tools and is also done in the appendix (see Section \ref{preuveannexe}).

%.................................... Début nouveau ..............................................
Finally, the last preliminary result is a corollary of a result by \cite{eagleson} on the stable-CLT for martingale increments adapted to our specific framework (in our case, the subset of concerned $\sigma$-fields are not ordered by the inclusion relation).
\begin{lemma}
\label{fstableconvergence} Let $(\Omega,{\cal F},\PE)$ denote a probability space. For $n\ge1$, let $\zeta_2^n, \zeta_3^n,\ldots,\zeta_{k_n}^n$ denote some martingale increments with respect to the sub-$\sigma$-fields of ${\cal F}$ $\bar{{\cal  F}}_{n,1}\subset\bar{{\cal  F}}_{n,2}\subset\ldots\subset\bar{{\cal  F}}_{n,k_n}$. Set  $S_n=\sum_{i=2}^{k_n} \zeta_i^n$ and ${\cal G}=\cap_{n\ge1}\bar{{\cal  F}}_{n,1}$.  Assume that  $n\rightarrow \bar{{\cal  F}}_{n,k_n}$ is a non-increasing sequence of $\sigma$-fields such that $\cap_{n\ge1}\bar{{\cal  F}}_{n,k_n}={\cal G}$.  Then, if the following conditions hold:
\begin{itemize}
\item[(i)]  There exists a ${\cal G}$-measurable random variable $\eta$ such that
\begin{equation}\label{cond1-tclmart}
\sum_{i=2}^{k_n}\ES[(\zeta_i^n)^2/\bar{{\cal  F}}_{n,{i-1}}]\xrn{\PE}\eta\quad \textnormal{as $n\rightarrow+\infty$},
\end{equation}
\item[(ii)] For every $\varepsilon>0$,
\begin{equation}\label{cond2-tclmart}
\sum_{i=2}^{k_n}\ES[(\zeta_i^n)^2{\bf 1}_{|\zeta_i^n|^2\ge \varepsilon}/\bar{{\cal  F}}_{n,{i-1}}]\xrn{\PE}0\quad \textnormal{as $n\rightarrow+\infty$},
\end{equation}
\end{itemize}
then, $(S_n)$ converges stably to $S$ where $S$ is defined on an extension $(\tilde{\Omega},\tilde{\cal F},\tilde{\PE})$ and such that conditionally on ${\cal F}$, the distribution of $S$ is a centered Gaussian law with variance $\eta$.
\end{lemma}
\begin{proof} First, we prove the lemma when the convergence in \eqref{cond1-tclmart} and \eqref{cond2-tclmart} holds $a.s.$ Then, the random variable $\eta$ being $\cap_{n\ge1}\bar{{\cal  F}}_{n,1}$-measurable, we deduce  from Corollary 2 of \cite{eagleson} that for every bounded ${\cal G}$-measurable random variable $Z$, for every bounded continuous function $f$ 
 \begin{equation}\label{restricted-conv2}
\ES[Z f(S_n)]\xrightarrow{n\rightarrow+\infty}\int_{\mathbb{R}}\ES[Z f(\sqrt{\eta} u)]\frac{1}{\sqrt{2\pi}}\exp(-\frac{u^2}{2})du.
\end{equation}
Now, let $Y$ denote a bounded ${\cal F}$-measurable random variable. Since $S_n$ is $\bar{{\cal  F}}_{n,k_n}$-measurable,
 $$|\ES[Yf(S_n)]-\ES[\ES[Y/{\cal G}]f(S_n)]|=|\ES[(\ES[Y/\bar{{\cal  F}}_{n,k_n}]-\ES[Y/{\cal G}])f(S_n)]|\le C\ES[|\ES[Y/\bar{{\cal  F}}_{n,k_n}]-\ES[Y/{\cal G}]|].$$
Under the assumptions of the lemma, $n\rightarrow \bar{{\cal  F}}_{n,k_n}$ is a non-increasing sequence and ${\cal G}=\cap_{n\ge1}\bar{{\cal  F}}_{n,k_n}$. Thus, by the convergence theorem for reverse martingales, 
$$\ES[Y/\bar{{\cal  F}}_{n,k_n}]\xrn{n\nrn}\ES[Y/{\cal G}]\quad a.s.$$
The function $f$ and the random variable $Y$ being bounded, it follows from the dominated convergence Theorem that 
$$\ES[Yf(S_n)]-\ES[\ES[Y/{\cal G}]f(S_n)]\xrn{n\nrn}0.$$
Finally, since $\ES[Y/{\cal G}]$ is a bounded ${\cal G}$-measurable  random variable, we deduce  from \eqref{restricted-conv2} that for every  bounded continuous function $f$,
 \begin{equation*}
\ES[Y f(S_n)]\xrightarrow{n\rightarrow+\infty}\int\ES[\ES[Y/{\cal G}] f(\sqrt{\eta} u)]\frac{e^{-\frac{u^2}{2}}}{\sqrt{2\pi}}du=\int\ES[Y f(\sqrt{\eta} u)]\frac{e^{-\frac{u^2}{2}}}{\sqrt{2\pi}}du=\tilde{\ES}[Yf(S)],
\end{equation*}
where  $S$ is defined on an extension $(\tilde{\Omega},\tilde{\cal F},\tilde{\PE})$ and such that conditionally on ${\cal F}$, the distribution of $S$ is a centered Gaussian distribution with variance $\eta$. \\
Assume now that the convergence in \eqref{cond1-tclmart} and \eqref{cond2-tclmart} only holds in probability. Following carefully the preceding proof, we observe that we only have to prove that \eqref{restricted-conv2} still holds: let $Z$ denote a bounded ${\cal G}$-measurable random variable and let $f$ be a bounded continuous function. Then, $n\rightarrow\ES[Zf(S_n)]$ is a bounded sequence.  Let $(\ES[Zf(S_{n_k})])_{k}$ denote a convergent subsequence . Using that it is enough to assume  \eqref{cond2-tclmart} for a countable family $(\varepsilon_k)$ and the fact that the convergence in probability implies the $a.s$ convergence of a subsequence, it follows from a diagonalization procedure that there exists a subsequence $(m_k)$ of $(n_k)$ such that  \eqref{cond1-tclmart} and \eqref{cond2-tclmart} hold $a.s.$ Then, a second application of Corollary 2 of \cite{eagleson} yields for any subsequence $(m_k)$:
$$\lim_{k\rightarrow+\infty}\ES[Zf(S_{m_k})]=\int\ES[Z f(\sqrt{\eta} u)]\frac{e^{-\frac{u^2}{2}}}{\sqrt{2\pi}}du.$$
This concludes the proof of the lemma.

%................................................ Fin Nouveau ..............................................
\end{proof}
\subsection{The CLTs for the Brownian martingale terms}
In this section, we focus on the main terms of the decomposition which satisfy a central limit theorem. 
\begin{proposition} \label{L5} 
Assume that $\Delta_n=o(h_n)$ and $\mathbf{(SH)_2}$.\\

\noindent(i). Then, 
\begin{equation}
\label{CLT}
\rho_n\left(\Lambda_2^{n}(t)+\Lambda_4^n(t)\right) \xrightarrow[n \rightarrow +\infty]{{\cal L}-s} f(t,p)U,
\end{equation}
where $U\sim {\cal N}(0,1)$, $U$ is independent of ${\cal F}_{t}$ and
\begin{equation}
\label{ftp+rho}
(f^2(t,p),~\rho_n) =
\begin{cases}
\left(\varphi_1(p,t,\sigma),~\sqrt{r_n}\right)
&\textnormal{if $h_n=o(\sqrt{\Delta_n} )$},\\
\left(\beta^2\varphi_1(p,t,\sigma)+\varphi_2(p,t),~\frac{1}{\sqrt{h_n}} \right)
&\textnormal{if $\frac{\sqrt{\Delta_n}}{h_n}\rightarrow\beta \in \R^*_+$},\\
\left(\frac{1}{3}p^2(\sigma_t)^{2p-2}\|\eta\|^2(t),~\frac{1}{\sqrt{h_n}} \right)
&\textnormal{if $\frac{\sqrt{\Delta_n}}{h_n}\rightarrow 0$}.
\end{cases}
\end{equation}
(ii).
In  case of pure jump process, meaning we assume that $\eta_1=\eta_2=0$,
then, $\Lambda_4=0$ and, for every $t\in[0,T]$, 
\begin{equation}
\label{CLT2}
\sqrt{\frac{h_n}{\Delta_n}}\Lambda_2^{n}(t) \xrightarrow[n \rightarrow +\infty]{{\cal L}-s} f(t,p)U,
\end{equation}
with $f^2(t,p)=\varphi_1(p,t,\sigma).$
\end{proposition}

\begin{proof}
Actually, in case (ii), the proof is easier since it only deals with $\Lambda_2^n,$ and is more or less included in what follows.\\

%........................................ NOUVEAU Rédaction et Notations à revoir.......................\\

\noindent Let $t>0$. Let 
$
\{(\xi_i^n),i=[t/\Delta_n]+2,\ldots,[(t+h_n)/\Delta_n],n\ge1\}
$
be the sequence  of martingale increments defined by: $\xi^n=\xi_i^{n,1}+\xi_i^{n,2}$ with
\begin{align}
\xi_i^{n,1}
&:=\frac{\rho_n}{r_n}\left( \left| \sigma_{(i-1) \Delta_n} \frac{\Delta^n_i W}{\sqrt{\Delta_n}} \right|^p - \espcond{i-1}{ \left| \sigma_{(i-1) \Delta_n} \frac{\Delta^n_i W}{\sqrt{\Delta_n}} \right|^p}\right),  \label{xi_ni}\\
\xi_i^{n,2}
&:=\frac{\rho_n}{r_n}([(t+h_n)/\Delta_n]-i+1) (M_{i\Delta_n}-M_{(i-1)\Delta_n\vee t}). \label{xi_ni2}
\end{align}
We first notice that  $\sum_{i=[\frac{t}{\Delta_n}]+2}^{[\frac{t+h_n}{\Delta_n}]}(\xi_i^{n,1}+\xi_i^{n,2})=\rho_n(\Lambda_2^n(t)+\Lambda_4^n(t))-\varepsilon_n$ where
$\varepsilon_n=\xi_{[\frac{t}{\Delta_n}]+1}^n\rightarrow0$ in probability. Second, let us show that Lemma \ref{fstableconvergence} can be applied to the sequence $(\xi_i^n)$. Set $k_n:=[(t+h_n)/\Delta_n]-[t/\Delta_n]$. For every $i\in\{2,\ldots,k_n\}$,
set $\zeta_{i,n}=\xi_{[t/\Delta_n]+i}$ which is $(\bar{\cal F}_{n,i})$-adapted where $\bar{\cal F}_{n,i}:={\cal F}_{(i+[t/\Delta_n])\Delta_n}$ for every $i\in\{1,\ldots,k_n-1\}$ and $\bar{\cal F}_{n,k_n}:={\cal F}_{t+h_n}$. We observe that the sequence $n\rightarrow \bar{\cal F}_{n,k_n}={\cal F}_{t+h_n}$ is nonincreasing  and that $\cap_{n\ge1} \bar{\cal F}_{n,1}= \cap_{n\ge1} \bar{\cal F}_{n,k_n}={\cal F}_t$ since the filtration $({\cal F}_t)$ is right-continuous. Then, we deduce from Lemma \ref{fstableconvergence} that in order to prove the proposition, it is now enough to check Conditions \eqref{cond1-tclmart} and \eqref{cond2-tclmart}. These conditions will follow from the two following lemmas.
%The following lemma gives the asymptotic predictable bracket of this sum of martingale increments. 
%.................................................. Fin...............................................
\begin{lemma} \label{var}
Let $f(t,p)$ defined by (\ref{ftp+rho}), then
$$
\sum_{i=[\frac{t}{\Delta_n}]+2}^{[\frac{t+h_n}{\Delta_n}]}\espcond{i-1}{(\xi_i^{n,1}+\xi_i^{n,2})^2} \xrightarrow[n\rightarrow+\infty]{\PRB} f^2(t,p).
$$ 
\end{lemma}

\begin{proof}
Three sums have to be computed:
$$
\sum_{i=[\frac{t}{\Delta_n}]+2}^{[\frac{t+h_n}{\Delta_n}]} \espcond{i-1}{(\xi_i^{n,1})^2}, \qquad
\sum_{i=[\frac{t}{\Delta_n}]+2}^{[\frac{t+h_n}{\Delta_n}]} \espcond{i-1}{(\xi_i^{n,2})^2}, \qquad
\sum_{i=[\frac{t}{\Delta_n}]+2}^{[\frac{t+h_n}{\Delta_n}]} \espcond{i-1}{\xi_i^{n,1} \xi_i^{n,2}}.
$$
(i) First
\begin{equation}  \label{**}
\espcond{i-1}{\left(\xi_i^{n,1}\right)^2}=\left(\frac{\rho_n}{r_n}\right)^2(m_{2p}-m_p^2) (\sigma_{(i-1) \Delta_n})^{2p},
\end{equation}
and since $\sigma$ is c\`ad,
$$
\frac{1}{r_n}\sum_{i=[\frac{t}{\Delta_n}]+2}^{[\frac{t+h_n}{\Delta_n}]} \sigma_{(i-1) \Delta_n}^{2p}\xrightarrow[n\rightarrow+\infty]{} \sigma_t^{2p}\quad a.s.
$$
Thus, by the definition of $\rho_n$,
\begin{equation}\label{crochet1}
\sum_{i=[\frac{t}{\Delta_n}]+2}^{[\frac{t+h_n}{\Delta_n}]}\espcond{i-1}{(\xi_i^{n,1})^2}\xrightarrow[n\rightarrow+\infty]{\PRB}
\begin{cases}
(m_{2p}-m_p^2)\sigma_t^{2p}&\textnormal{if $h_n=o(\sqrt{\Delta_n})$},\\
\beta^2(m_{2p}-m_p^2)\sigma_t^{2p}&\textnormal{if $\sqrt{\Delta_n}/h_n\rightarrow\beta\in\R_+^*$},\\
0&\textnormal{if $\sqrt{\Delta_n}/h_n\rightarrow 0$}.
\end{cases}
\end{equation}
(ii) Second,
\begin{align*}
&\espcond{i-1}{\left(\xi_i^{n,2}\right)^2}=\left(\frac{\rho_n([(t+h_n)/\Delta_n]-i+1)}{r_n}\right)^2
\int_{\CNT}\espcond{i-1}{\psi_s} ds,
\end{align*}
with $\psi_s=p^2|Y_s|^{2p-2}[\eta_1^2(s)+\eta_2^2(s)]$. One observes that
\begin{align*}
&\esp{ \sum_{i=[\frac{t}{\Delta_n}]+2}^{[\frac{t+h_n}{\Delta_n}]}\left(\frac{\rho_n([(t+h_n)/\Delta_n]-i+1)}{r_n}\right)^2
\left|\int_{\CNT}(\espcond{i-1}{\psi_s}-\psi_t) ds\right|},\\
&\qquad\le \Delta_n\sum_{i=[\frac{t}{\Delta_n}]+2}^{[\frac{t+h_n}{\Delta_n}]}\left(\frac{\rho_n([(t+h_n)/\Delta_n]-i+1)}{r_n}\right)^2\esp{\sup_{s\in[t, t+h_n]}|\psi_s-\psi_t|},\\
&\qquad\le C\rho_n^2 h_n \esp{\sup_{s\in[t, t+h_n]}|\psi_s-\psi_t|}.
\end{align*}
The function $\psi$ is c\`ad. Therefore, using  \eqref{EspSup} and the fact that $\eta_1$ and $\eta_2$ are bounded, we deduce from the dominated convergence theorem that for every $t\in[0,T]$,
\begin{equation} \label{argumentrecurrent}
\esp{\sup_{s\in[t, t+h_n]}|\psi_s-\psi_t|} \xrightarrow[n\rightarrow+\infty]{} 0.
\end{equation}
It follows from the definition of $\rho_n$ that 
$$
\sum_{i=[\frac{t}{\Delta_n}]+2}^{[\frac{t+h_n}{\Delta_n}]}\left(\espcond{i-1}{(\xi_i^{n,2})^2}-\left(\frac{\rho_n([(t+h_n)/\Delta_n]-i+1)}{r_n}\right)^2\Delta_n\psi_t\right)\xrightarrow[n\rightarrow+\infty]{\PRB}0.
$$
Thus, since 
$$
\frac{1}{h_n}\sum_{i=[\frac{t}{\Delta_n}]+2}^{[\frac{t+h_n}{\Delta_n}]}\left(\frac{([(t+h_n)/\Delta_n]-i+1)}{r_n}\right)^2\Delta_n\xrightarrow[n\rightarrow+\infty]{}\frac{1}{3},
$$
we obtain that the order of $\sum_{i=[\frac{t}{\Delta_n}]+2}^{[\frac{t+h_n}{\Delta_n}]}\espcond{i-1}{(\xi_i^{n,2})^2}$
is $\frac{1}{3}\rho_n^2h_n\psi_t$ thus
\begin{equation} \label{crochet2}
\sum_{i=[\frac{t}{\Delta_n}]+2}^{[\frac{t+h_n}{\Delta_n}]}\espcond{i-1}{(\xi_i^{n,2})^2} \xrightarrow[n\rightarrow+\infty]{\PRB}
\begin{cases}
0&\textnormal{if $h_n=o(\sqrt{\Delta_n})$}\\
\frac{\psi_t}{3}&\textnormal{if $\sqrt{\Delta_n}/h_n\rightarrow\beta\in\R_+$.}
\end{cases}
\end{equation}
(iii) Finally, we consider the cross products $\espcond{i-1}{\xi_i^{n,1}\xi_i^{n,2}}$. First of all, it is easily seen that, $W$ and $W^2$ being independent, only the term in $W$ of $M$ will play a role. Thus we have:
\begin{align*}
&\espcond{i-1}{\xi_i^{n,1}\xi_i^{n,2}}  \\
&= \alpha_{i,n}(t) \sigma_{(i-1)\Delta_n}^p \espcond{i-1}{\int_{\CNT} p \sigma_s^{p-1}\eta_1(s)dW_s \left(|\Delta_i^n W|^p - \espcond{i-1}{|\Delta_i^n W|^p}\right)}
\end{align*}
with $\alpha_{i,n}(t)=(\rho_n/r_n)^2\Delta_n^{-p/2}([(t+h_n)/\Delta_n]-i+1)$. Now, by It\^o's formula, 
$$
\left|\Delta_i^nW\right|^p
=p\int_{\CN}~{\rm sgn}(W_s-W_{(i-1)\Delta_n})|W_s-W_{(i-1)\Delta_n}|^{p-1}dW_s+\frac{p(p-1)}{2}\int_{\CN}|W_s-W_{(i-1)\Delta_n}|^{p-2}ds.
$$
Then, we have $\espcond{i-1}{\xi_i^{n,1}\xi_i^{n,2}}=T_i^{n,1}+T_i^{n,2}$ with
\begin{align*}
T_i^{n,1} &=p^2\alpha_{i,n}(t)\sigma_{(i-1)\Delta_n}^p\int_{\CNT}\espcond{i-1}{\sigma_s^{p-1}\eta_1(s)~{\rm sgn}(W_s-W_{(i-1)\Delta_n})|W_s-W_{(i-1)\Delta_n}|^{p-1}}ds,\\
T_i^{n,2} &=\frac{p^2(p-1)}{2}\alpha_{i,n}(t)\sigma_{(i-1)\Delta_n}^p      \espcond{i-1}{\int_{\CN}|W_s-W_{(i-1)\Delta_n}|^{p-2}ds\int_{\CNT}\sigma_s^{p-1}\eta_1(s)dW_s}.
\end{align*}

\noindent First, let us focus on $T_i^{n,2}$. By an integration by parts, one obtains that: 
\begin{align*}
&\espcond{i-1}{\int_{\CNT}|W_s-W_{(i-1)\Delta_n\vee t}|^{p-2}ds\int_{\CNT}\sigma_s^{p-1}\eta_1(s)dW_s}\\
&=\int_{\CNT}\espcond{i-1}{\left(\int_{\CNTS} \sigma_u^{p-1}\eta_1(u)dW_u\right)|W_s-W_{(i-1)\Delta_n\vee t}|^{p-2}}ds\\
&=\int_{\CNT}
\espcond{i-1}{\int_{\CNTS} (\sigma_u^{p-1}\eta_1(u)-\sigma_t^{p-1}\eta_1(t))dW_u.|W_s-W_{(i-1)\Delta_n\vee t}|^{p-2}}ds,
\end{align*}
where in the last line we used that for every $s \in [(i-1)\Delta_n,i\Delta_n]$, 
$$
\espcond{i-1}{(W_s-W_{(i-1)\Delta_n\vee t})|W_s-W_{(i-1)\Delta_n\vee t}|^{p-2}}=0.
$$
Then, using Cauchy-Schwarz inequality 
\begin{align*}
&\espcond{i-1}{\left(\int_{\CNTS}\left(\sigma_u^{p-1}\eta_1(u)-\sigma_t^{p-1}\eta_1(t)\right)dW_u\right)|W_s-W_{(i-1)\Delta_n\vee t}|^{p-2}}\\
&\quad \leq 
\sqrt{\espcond{i-1}{\left(\int_{\CNTS} \left(\sigma_u^{p-1}\eta_1(u)-\sigma_t^{p-1}\eta_1(t)\right)dW_u\right)^2}}
.~\sqrt{\espcond{i-1}{|W_s-W_{(i-1)\Delta_n\vee t}|^{2p-4}}},\\
&\quad \leq 
\sqrt{\int_{\CNTS} \espcond{i-1}{\left(
\sup_{u\in[t,t+h_n]}\left|\sigma_u^{p-1}\eta_1(u)-\sigma_t^{p-1}\eta_1(t)\right|^2\right)}du}
.~(s-(i-1)\Delta_n)^{(p-2)/2}.
\end{align*}
Then \eqref{EspSup} yields:
\begin{align*}
\esp{|T_i^{n,2}|}&\le C\alpha_{i,n}(t)\esp{\sup_{u\in[t,t+h_n]}\left|\sigma_u^{p-1}\eta_1(u)-\sigma_t^{p-1}\eta_1(t)\right|^2}^{\frac{ 1}{2}} \int_{\CNT}(s-(i-1)\Delta_n)^{\frac{p-1}{2}}ds,\\
&\le C\frac{\rho_n^2\sqrt{\Delta_n}}{r_n^2}([(t+h_n)/\Delta_n]-i+1) \esp{\sup_{u\in[t,t+h_n]}\left|\sigma_u^{p-1}\eta_1(u)-\sigma_t^{p-1}\eta_1(t)\right|^2}^\demi.
\end{align*}
Thus, an argument similar to \eqref{argumentrecurrent} yields:
\begin{equation} \label{termeneg}
\sum_{i=[\frac{t}{\Delta_n}]+2}^{[\frac{t+h_n}{\Delta_n}]}~~T_i^{n,2}\xrightarrow[n\rightarrow+\infty]{\PRB}0.
\end{equation}
Second, we focus on $T_i^{n,1}$. Using again that $\sigma$ and $\eta$ are c\`ad, one obtains that
$$
\sum _{i=[\frac{t}{\Delta_n}]+2}^{[\frac{t+h_n}{\Delta_n}]} \left[T_i^{n,1}-p^2\alpha_{i,n}(t)
\sigma_t^{2p-1}\eta_1(t)\int_{\CNT}
\espcond{i-1}{~{\rm sgn}(W_s-W_{(i-1)\Delta_n})|W_s-W_{(i-1)\Delta_n}|^{p-1}}ds\right]
\xrightarrow[n\rightarrow+\infty]{\PRB} 0.
$$
Then, since $\espcond{i-1}{~{\rm sgn}(W_s-W_{(i-1)\Delta_n})|W_s-W_{(i-1)\Delta_n}|^{p-1}}=0$
%$$
%\frac{1}{\sqrt{\Delta_n}}\sum_{i=[\frac{t}{\Delta_n}]+2}^{[\frac{t+h_n}{\Delta_n}]}
%\frac{([(t+h_n)/\Delta_n]-i+1)}{r_n^2\Delta_n^{p/2}}\int_{\CNT}(s-(i-1)\Delta_n)^{\frac{p-1}{2}}ds
%\xrightarrow[n\rightarrow+\infty]{}\frac{1}{2},
%$$
we deduce that
$$\sum_{i=[\frac{t}{\Delta_n}]+2}^{[\frac{t+h_n}{\Delta_n}]}~~ T_i^{n,1}
\xrightarrow[n\rightarrow+\infty]{\PRB} 0.$$
%\begin{equation*}
%\sum_{i=[\frac{t}{\Delta_n}]+2}^{[\frac{t+h_n}{\Delta_n}]}~~ T_i^{n,2}
%\xrightarrow[n\rightarrow+\infty]{\PRB}
%\begin{cases}
%0&\textnormal{if $h_n=o(\sqrt{\Delta_n})$},\\
%\beta \frac{p^2 \sigma_t^{2p-1}\eta_1(t)}{2}&\textnormal{if $\sqrt{\Delta_n}/h_n\rightarrow\beta\in\R_+^*$},\\
%0&\textnormal{if $\sqrt{\Delta_n}/h_n\rightarrow 0$.}
%\end{cases}
%\end{equation*}
then with \eqref{termeneg} that 

\begin{equation}
\label{crochet3}
\sum_{i=[\frac{t}{\Delta_n}]+2}^{[\frac{t+h_n}{\Delta_n}]}\espcond{i-1}{\xi_i^{n,1}\xi_i^{n,2}}\xrightarrow[n\rightarrow+\infty]{\PRB}0.
\end{equation}

%\begin{equation}\label{crochet3}
%\sum_{i=[\frac{t}{\Delta_n}]+2}^{[\frac{t+h_n}{\Delta_n}]}\espcond{i-1}{\xi_i^{n,1}\xi_i^{n,2}}\xrightarrow[n\rightarrow+\infty]{\PRB}
%\begin{cases}
%0&\textnormal{if $h_n=o(\sqrt{\Delta_n})$},\\
%\beta\frac{p^2 \sigma_t^{2p-1}\eta_1(t)}{2}&\textnormal{if $\sqrt{\Delta_n}/h_n\rightarrow\beta\in\R_+^*$},\\
%0&\textnormal{if $\sqrt{\Delta_n}/h_n\rightarrow 0$.}
%\end{cases}
%\end{equation}
Thus, by \eqref{crochet1}, \eqref{crochet2} and \eqref{crochet3}, we obtain that,
$$
\sum_{i=[\frac{t}{\Delta_n}]+2}^{[\frac{t+h_n}{\Delta_n}]}\espcond{i-1}{(\xi_i^{n,1}+\xi_i^{n,2})^2} \xrightarrow[n\rightarrow+\infty]{\PRB} f^2(t,p).
$$ 
\end{proof}

%Owing to  Lemma \ref{fstableconvergence}, the following lemma will conclude the proof of Proposition \ref{L5}.
%Corollary 2 of \cite{eagleson} (see also \cite{hall}, p. 59).
\begin{lemma}
The following Lindeberg condition holds:
\begin{equation}\label{lindeberg}
\sum_{i=[\frac{t}{\Delta_n}]+2}^{[\frac{t+h_n}{\Delta_n}]}\espcond{i-1}{(\xi_i^{n,1}+\xi_i^{n,2})^2 \1_{|\xi_i^{n,1}+\xi_{i}^{n,2}|^2\ge \varepsilon}} \xrightarrow[n \rightarrow +\infty]{} 0\quad a.s. \quad\forall\varepsilon>0.
\end{equation}
\end{lemma}

\begin{proof}
Let us prove \eqref{lindeberg}. We derive from the Cauchy-Schwarz and Chebyshev inequalities that,
\begin{align*}
\espcond{i-1}{\left(\xi_i^{n,1}+\xi_i^{n,2}\right)^2 \1_{\left\{|\xi_i^{n,1}+\xi_i^{n,2}|^2\ge \varepsilon \right\} } }& \le \espcond{i-1}{\left(\xi_i^{n,1}+\xi_i^{n,2}\right)^4}^{\frac{1}{2}} \left[ \pr{}{ \left\{ |\xi_i^{n,1}+\xi_i^{n,2}|^2\ge \varepsilon \right\}~\Big|~{\cal F}_{(i-1)\Delta_n}}\right]^{\frac{1}{2}},\\
& \le \frac{8}{\varepsilon}\left(\espcond{i-1}{\left(\xi_i^{n,1}\right)^4}+\espcond{i-1}{\left(\xi_i^{n,2}\right)^4}\right).
\end{align*}
On the one hand, using \eqref{xi_ni},
$$
\espcond{i-1}{\left(\xi_i^{n,1}\right)^4}= \frac{\rho_n^4}{r_n^4}\sigma_{(i-1) \Delta_n}^{4p} \esp{(|U|^p-m_p)^4},
$$
and since $\sigma$ is locally bounded, we obtain that there exists $C(\omega)$ such that for all $t \ge 0$,
$$
\sum_{i=[\frac{t}{\Delta_n}]+2}^{[\frac{t+h_n}{\Delta_n}]}\espcond{i-1}{(\xi_i^{n,1})^4 }
\le \frac{C(\omega)}{\varepsilon}\sum_{i=[\frac{t}{\Delta_n}]+2}^{[\frac{t+h_n}{\Delta_n}]}\frac{\rho_n^4}{r_n^4}
= \frac{C(\omega)}{\varepsilon}\frac{\rho_n^4}{r_n^3}.
$$
If $h_n=o(\sqrt{\Delta_n})$ (resp. $\sqrt{\Delta_n}=O(h_n)$), $\rho_n^4/r_n^3=1/r_n$ (resp. $\rho_n^4/r_n^3=\Delta_n^{3/2}/h_n^5$). Thus,\\
$\sum_{i=[\frac{t}{\Delta_n}]+2}^{[\frac{t+h_n}{\Delta_n}]}\espcond{i-1}{(\xi_i^{n,1})^4 }\xrightarrow[n\rightarrow+\infty]{}0\quad a.s.$\\
On the other hand, using (\ref{xi_ni2})
\begin{align*}
\esp{\sum_{i=[\frac{t}{\Delta_n}]+2}^{[\frac{t+h_n}{\Delta_n}]}\espcond{i-1}{\left(\xi_i^{n,2}\right)^4}}&=\left(\frac{\rho_n}{r_n}\right)^4\left([(t+h_n)/\Delta_n]-i+1\right)^4\esp{(M_{i\Delta_n}-M_{(i-1)\Delta_n\vee t})^4}\\
&\quad\le\sum_{i=[\frac{t}{\Delta_n}]+2}^{[\frac{t+h_n}{\Delta_n}]}\left(\frac{\rho_n}{r_n}\right)^4\left([(t+h_n)/\Delta_n]-i+1\right)^4\esp{\left(\int_{\CNT}\psi(s)ds\right)^2}.\\
&\quad\le \sum_{i=[\frac{t}{\Delta_n}]+2}^{[\frac{t+h_n}{\Delta_n}]}\left(\frac{\rho_n}{r_n}\right)^4\left([(t+h_n)/\Delta_n]-i+1\right)^4\Delta_n \int_{\CNT}\esp{\psi(s)^2}ds.
\end{align*}
Since $\psi(s)\le C|\sigma_t|^{2p-2}$, it follows from \eqref{EspSup} that $\sup_{s\in[0,T]}\esp{\psi(s)^2}<+\infty$. Now,

$$\sum_{i=[\frac{t}{\Delta_n}]+2}^{[\frac{t+h_n}{\Delta_n}]}\left(\frac{\rho_n}{r_n}\right)^4\left([(t+h_n)/\Delta_n]-i+1\right)^4\Delta_n^2\le
C\rho_n^4 r_n\Delta_n^2,$$
and one checks that this right-hand member tends to 0 in every cases.
It follows that the Lindeberg condition is fulfilled.
\end{proof}
\noindent These two lemmas conclude the proof of Proposition \ref{L5}.
\end{proof}

\subsection{The remainder terms}
 We focus on $\Lambda_5^n(t),$  recalling:
 $$
 \Lambda_5^n(t)=\frac{1}{r_n} \sum_{i\in\DNT}([(t+h_n)/\Delta_n]-i+1)(A_{i\Delta_n}-A_{(i-1)\Delta_n\vee t}).$$
 We obtain the following results of convergence in probability.
 \begin{proposition}
 \label{lambda5}
 Assume $\mathbf{(SH)_q}$ and $\mathbf{(H^2_q)}$ with $q\in]1,2]$. Then, for every $t\in[0,T]$:
\begin{equation}
\label{sigma5}
\frac{1}{h_n^{1/q}}\Lambda_5^{n}(t)\xrightarrow[n\rightarrow+\infty]{\PRB}0.
\end{equation}
Assume that the previous assumptions hold with  $q=1$ and that $\left(\int_{\{0<|y|\le1\}}y F_t(dy)\right)_{t\ge0}$ is càglàd. Let $(\theta_t^0)$ be defined by 
\begin{equation}\label{thetato}
\theta_t^0:= p\sigma_t^{p-1}\left(b_t-\lim_{s\searrow t}\int_{\{0<|y|\le1\}}y F_s(dy)\right)+\frac{p(p-1)}{2}\sigma_t^{p-2}\|\eta\|^2(t),
\end{equation}
Then for every $t\in[0,T]$,
\begin{equation}
\label{sigma5'}
\frac{1}{h_n}\Lambda_5^{n}(t)\xrightarrow[n\rightarrow+\infty]{\PRB} \frac{\theta_t^0}{2}.
\end{equation}
 \end{proposition}
 
\begin{remark} Note that Assumption $\mathbf{(H_q^2)}$ is only necessary at this stage of the proof where a kind of regularity of the small jumps is needed.
\end{remark}
 
\begin{proof}
It will be useful to notice that, for every $\epsilon > 0$, 
$
A_t = \cLe(t) + \cMe(t) + \cNe(t),
$
with
\begin{align*}
\cLe(t) &=\int_0^t \left(\theta_s - \int_{\epsilon \leq |y| \leq 1}  \,p|Y_{s}|^{p-1} y F_s(dy)\right)ds,\\
\cMe(t) &=\int_0^t \int_{|y| \leq \epsilon}  p|Y_{s^-}|^{p-1} y({\mu}-{\nu})(ds,dy)+\sum_{0< s \le t} H^{g_p}(Y_{s-},\Delta Y_s ) \1_{[|\Delta Y_s| \leq \epsilon]},\\
\cNe(t) &= \sum_{0< s \le t} \left( |Y_{s-} + \Delta Y_s |^p - |Y_{s-}|^p \right) \1_{\{|\Delta Y_s| > \epsilon\}} .
\end{align*}
where $g_p(x)=|x|^p$ and for every $f:\R\rightarrow\R$,
 \begin{equation}\label{HF}
 H^f(x,y)=f(x+y)-f(x)-f'(x)y.
 \end{equation}
With these notations, the above proposition is a consequence of the following lemma.
\begin{lemma} 
\label{Lambda5}
Assume $\mathbf{(SH)_q}$ with $q\in[1,2]$. Then,\\
(i) For every $\varepsilon>0$, there exists $a.s.$ $n_0(\omega)$ such that for every $n\ge n_0(\omega)$,
\begin{equation*}
\frac{1}{r_n} \sum_{i\in\DNT} \left( \left[ \frac{t+h_n}{\Delta_n}\right]-i+1\right) ({\cal N}^\varepsilon_{i\Delta_n}-{\cal N}^\varepsilon_{(i-1)\Delta_n\vee t})=0.
\end{equation*}
(ii)  Assume moreover $\mathbf{(H^2_q)}$ with $q\in[1,2]$.
 For every $\delta>0$, there exists, $\varepsilon_\delta>0$ such that for every $\varepsilon\le \varepsilon_\delta$: 
\begin{equation*}
\pr{}{\frac{1}{r_n} \Big|\sum_{i\in\DNT} \left( \left[ \frac{t+h_n}{\Delta_n}\right]-i+1\right) ({\cal M}^\varepsilon_{i\Delta_n}-{\cal M}^\varepsilon_{(i-1)\Delta_n\vee t})\Big|> \delta h_n^{1/q}} \xrightarrow[n\rightarrow+\infty]{} 0.
\end{equation*}
(iii) For every $\varepsilon>0$, we have almost surely
\begin{equation}\label{sigma3}
\limsup_{n\rightarrow+\infty}\frac{1}{h_n}\left(\frac{1}{r_n} \sum_{i\in\DNT}\left( \left[ \frac{t+h_n}{\Delta_n}\right]-i+1\right) \VA{{\cal L}^\varepsilon_{i\Delta_n}-{\cal L}^\varepsilon_{(i-1)\Delta_n\vee t}}\right)<+\infty.
\end{equation}
Assume moreover that $\mathbf{(SH)_1}$ and $\mathbf{(H_1^2)}$ hold  and
  that $\left(\int_{\{0<|y|\le1\}}y F_t(dy)\right)_{t\ge0}$ is càglàd. Then, almost surely,
\begin{equation}\label{sigma4}
\limsup_{\varepsilon\rightarrow 0}\limsup_{n\rightarrow+\infty}
\frac{1}{h_n}~\frac{1}{r_n} \sum_{i\in\DNT}\left( \left[ \frac{t+h_n}{\Delta_n}\right]-i+1\right) \left({\cal L}^\varepsilon_{i\Delta_n}-{\cal L}^\varepsilon_{(i-1)\Delta_n\vee t}\right)= \frac{\theta_t^0}{2}.
\end{equation}
\end{lemma}

\begin{proof} $(i)$ Let $T_t^{\varepsilon}$ denote the random time defined by $T_t^{\varepsilon}(\omega):=\inf\{s> t, |\Delta Y_s|\ge \varepsilon\}$. For every $\delta>0$, 
$$\pr{}{ t\le T_t^{\varepsilon}\le t+\delta }\le\esp{\sum_{t\le s\le t+\delta} \1_{\{|\Delta Y_s|\ge \varepsilon\}}}\le \esp{\int_t^{t+\delta}\int_{\{|y|\ge\varepsilon\}} F_s(dy)ds}.$$
Under $\mathbf{(SH)_2}$, 
 $$\int_{\{|y|\ge\varepsilon\}} F_s(dy)\leq\eps^{-2}\int_{\{|y|\ge\varepsilon\}}|y|^2 F_s(dy)
\leq\eps^{-2}\sup_{s\in[0,T]}\int |y|^2 F_s(dy)\le M/\varepsilon^2.$$
It follows from the dominated convergence theorem that $\pr{}{T_t^{\varepsilon}=t}=0$. Thus, $a.s.$, there exists $n_0(\omega)$ such that $T_t(\omega)>t+h_n$ for every $n\ge n_0(\omega)$. The result follows.\\

\noindent $(ii)$ On the one hand, by the Doob inequality for discrete martingales, we have for every $q\in(1,2]$,
\begin{align*}
&\esp{\Big|\sum_{i\in\DNT}([(t+h_n)/\Delta_n]-i+1)\int_{\CNT}\int_{\{|y|\le\varepsilon\}} p(\sigma_{s^-})^{p-1} y(\mu-\nu)(ds,dy)\Big|^q}\\
&\quad\le\esp{\Big|\sum_{i\in\DNT}([(t+h_n)/\Delta_n]-i+1)^2\int_{\CNT}\int_{\{|y|\le\varepsilon\}} p^2\sigma_{s}^{2p-2} y^2\nu(ds,dy)\Big|^\frac{q}{2}}.
\end{align*}
Then, using that $(\sum |u_i|)^{q/2}\le \sum |u_i|^{q/2}$ (since $q/2\le1$) and Jensen's inequality, we obtain:
\begin{align*}
&\esp{\Big|\sum_{i\in\DNT}\left( \left[ \frac{t+h_n}{\Delta_n} \right] -i+1 \right)\int_{\CNT}\int_{\{|y|\le\varepsilon\}} p(\sigma_{s^-})^{p-1} y(\mu-\nu)(ds,dy)\Big|^q}\\
&\le C \sum_{i\in\DNT}\left( \left[ \frac{t+h_n}{\Delta_n} \right] -i+1 \right)^q
\Delta_n^{\frac{q}{2}-1}\int_{\CNT}\esp{\int_{\{|y|\le\varepsilon\}} \sigma_{s}^{(p-1)q} |y|^q F_s(dy)}ds.
\end{align*}
Using Assumption
$\mathbf{(SH)_q}$, we derive from Cauchy-Schwarz's inequality,  \eqref{EspSup},  Assumption $\mathbf{(H_q^2)}$  and the dominated convergence Theorem that:
\begin{equation*}
\sup_{s\in[0,T]}\esp{\int_{\{|y|\le\varepsilon\}} \sigma_{s}^{(p-1)q} |y|^q F_s(dy)}\xrightarrow[\varepsilon\rightarrow 0]{}0.
\end{equation*}
Thus, using that
$$\sum_{i\in\DNT}([(t+h_n)/\Delta_n]-i+1)^q\le C r_n^{q+1},$$
it follows that for every $q\in[1,2]$, for every $\eta>0$, there exists $\varepsilon_\eta^1>0$ such that for every $\varepsilon\le \varepsilon_\eta^1$,
\begin{equation}\label{345}
\esp{ \Big| \frac{1}{r_n} \sum_{i\in\DNT} \left( \left[ \frac{t+h_n}{\Delta_n} \right] -i+1 \right) \int_{\CNT} \int_{-\varepsilon }^{\varepsilon} p\sigma_{s^-}^{p-1} y(\mu-\nu)(ds,dy) \Big| } \le C\eta h_n^{1/q}.
\end{equation}
On the other hand, by the Taylor formula, we have 
$$|H^{g_p}(x,y)| \le C \left(|x|^{p-2}|y|^2+|y|^{2p}\right),$$
 when $p\ge2$. Then, using that $\mathbf{(H^2_q)}$ implies  $\mathbf{(H^2_2)}$ and $\mathbf{(H^2_{2p})}$, we obtain that for every $\eta>0$, there exists $\varepsilon_0>0$ such that for every $\varepsilon \le \varepsilon_\eta^2$,
$$\esp{ \sum_{s \in \CNT} |H^{g_p}(Y_s^{-},\Delta Y _s)| \1_{\{|\Delta Y_s| \le \varepsilon\}} }\le C \eta\Delta_n.$$
Thus, for every $\eta>0$, there exists $\varepsilon_\eta^2$ such that for every $\varepsilon\le \varepsilon_\eta^2$,
\begin{equation}\label{346}
\esp{ \Big|\frac{1}{r_n}\sum_{i\in\DNT}([(t+h_n)/\Delta_n]-i+1)H^{g_p}(Y_s^{-},\Delta Y _s) \1_{\{|\Delta Y_s|\le \varepsilon\}}\Big|}\le C\rho r_n\Delta_n.
\end{equation}
Therefore, \textit{(ii)} follows from \eqref{345} and \eqref{346}.\\
 
\noindent (iii) Since $\theta_s$, $Y_s$ and  $\int_{\epsilon \leq |y| \leq 1}   y F_s(dy)$ are locally bounded, there exists almost surely $C_T(\omega)$ such that for every $t\in[0,T]$, for every $n\ge1$, $|{\cal L}^\epsilon_{i\Delta_n}-{\cal L}^\epsilon_{(i-1)\Delta_n\vee t}|\le C_T(\omega)\Delta_n.$ Assertion \eqref{sigma3} follows. By construction and under the assumptions on $\sigma$ and $b$, $(\theta_s^0)$ is c\`ad. Then,
$$
{\cal L}^\epsilon_{i\Delta_n}-{\cal L}^\epsilon_{(i-1)\Delta_n\vee t}= \theta_t^0\Delta_n+R_i^n(\varepsilon,t)+o_{\omega,t}(\Delta_n),
$$
with $R_i^n(\varepsilon,t)=-p\sigma_t^{p-1}\int_{\CNT}\int_{|y|\le \varepsilon} y F_s(dy) ds.$ 
Then, since 
$$\frac{1}{h_n}\left(\frac{1}{r_n}\sum_{i\in\DNT}([(t+h_n)/\Delta_n]-i+1)\Delta_n\right)\xrightarrow[n\rightarrow+\infty]{}\frac{1}{2},$$
it follows that 
\begin{align*}
\limsup_{n\rightarrow+\infty}&\frac{1}{h_n}\Big|\frac{1}{r_n} \sum_{i\in\DNT}([(t+h_n)/\Delta_n]-i+1) \left({\cal L}^\varepsilon_{i\Delta_n}-{\cal L}^\varepsilon_{(i-1)\Delta_n\vee t}\right)- \frac{\theta_t^0}{2}\Big|\\
&\le \limsup_{n\rightarrow+\infty}\frac{1}{h_n r_n}\sum_{i\in\DNT}([(t+h_n)/\Delta_n]-i+1)|R_i^n(\varepsilon,t)|,\\
&\le C\sigma_t^{p-1}\sup_{s\in[0,T]}\int_{|y|\le \varepsilon}|y|F_s(dy).
\end{align*}
Finally,  we deduce \eqref{sigma4} from $\mathbf{(H_1^2)}$.
\end{proof}
\end{proof}

\begin{lemma} 
\label{L4}
Assume  $\mathbf{(SH)_2}$.
 Then, there exists $C_p>0$ such that for all~$t$
$$
\sup_{t \in [0,T]} \, \esp{ (\Lambda_1^{n}(t))^2 } \leq C_p \frac{\Delta_n^{1+\frac{1}{p}}}{h_n}.
$$
As a consequence, for every $t\in[0,\bar{T}]$, 
$$
\sqrt{\frac{h_n}{\Delta_n}} \Lambda_1^{n}(t) \xrightarrow[n \rightarrow +\infty]{\bL^2} 0.
$$
\end{lemma}

\begin{proof}
Set $\sigma_{i}^n=\sigma_{i\Delta_n}$. Then, by a martingale argument, we have 
\begin{align*}
\esp{(\Lambda_1^n(t))^2}
&\leq \frac{\Delta_n^2}{h_n^2} \sum_{i\in\DNT}\esp{ \left(\VA{\frac{\Delta_i^n X}{\sqrt{\Delta_n}}}^p- \VA{\sigma_{i-1}^n \frac{\Delta^n_i W}{\sqrt{\Delta_n}}}^p \right)^2}, 
\\
&\leq \frac{\Delta_n^{2-p}}{h_n^2} \sum_{i\in\DNT}
\esp{ \left( \left| \Delta_i^n X \right|^p - \left| \sigma_{i-1}^n \Delta^n_i W \right|^p \right)^2}.
\end{align*}
As $dX_t=a_tdt+\sigma_tdW_t$, we have $\Delta_i^n X = \sigma_{(i-1) \Delta_n} \Delta^n_i W + \chi_i^n,$ with
$$
\chi_i^n = \int_{\CN} (\sigma_s - \sigma_{(i-1) \Delta_n}) dW_s + \int_{\CN} a_sds.
$$
Using a Taylor expansion of $g(x) = |x|^p$ on the interval $[\sigma_{(i-1) \Delta_n} \Delta_i^n W ; \Delta_i^n X]$, we have :
\begin{equation*}
\left| |\Delta_i^n X|^p - |\sigma_{(i-1) \Delta_n} \Delta_i^n W|^p \right|
\leq \underset{x \in [\sigma_{(i-1) \Delta_n} \Delta_i^n W ; \Delta_i^n X]}{\sup} |g'(x)| ~ | \chi_i^n |.
\end{equation*}
But $|g'(x)| = O(|x|^{p-1})$ thus using the relation $|x+y|^p \leq C_p(|x|^p + |y|^p)$ with $C_p$ a constant, we have
\begin{align*}
\underset{x \in [\sigma_{(i-1) \Delta_n} \Delta_i^n W ; \Delta_i^n X]}{\sup} |g'(x)| ~
&\leq~ C_p(| \sigma_{(i-1) \Delta_n} \Delta_i^n W  |^{p-1} ~+~ | \chi_i^n |^{p-1}),\\
\left| |\Delta_i^n X|^p - |\sigma_{(i-1) \Delta_n} \Delta_i^n W|^p \right|
&\leq C_p(| \sigma_{(i-1) \Delta_n} \Delta_i^n W  |^{p-1} | \chi_i^n | ~+~ | \chi_i^n |^{p}).
\end{align*}
Finally there is a constant $C_p$ such that, for all $t \ge 0$:
\begin{align} 
\label{E-LOC1}
&\esp{(\Lambda_1^n(t))^2} \notag\\
&\leq C_p \frac{\Delta_n^{2-p}}{h_n^2} \sum_{i\in\DNT} 
\esp{ | \chi_i^n |^2 | \sigma_{(i-1) \Delta_n} \Delta_i^n W  |^{2p-2} ~+~ | \chi_i^n |^{2p} }, \notag\\
&\leq C_p \frac{\Delta_n^{2-p}}{h_n^2} \sum_{i\in\DNT} 
\left[\left(\esp{ | \chi_i^n |^{2p}}\right)^{\frac{1}{p}} \left(\esp{ | \sigma_{(i-1) \Delta_n} \Delta_i^n W  |^{2p}}\right)^{\frac{p-1}{p}} + \esp{| \chi_i^n |^{2p} } \right].
\end{align}
First of all, the independence between $\sigma_{(i-1) \Delta_n}$ and $\Delta_i^n W $ and (\ref{EspSup}) yield:
\begin{equation*} %\label{E-AVOIR1}
\esp{ | \sigma_{(i-1) \Delta_n} \Delta_i^n W  |^{2p}} 
= \Delta_n^{p} \, m_{2p} \, \esp{| \sigma_{(i-1) \Delta_n} |^{2p}} 
\leq C_p \Delta_n^{p}.
\end{equation*}
So it remains to give a majoration of $\esp{| \chi_i^n |^{2p} }$. Since $a$ is bounded by $M$,
\begin{equation*}
\esp{| \chi_i^n |^{2p}} 
\leq C_p.\left(
\esp{ \VA{ \int_{\CN} (\sigma_s - \sigma_{(i-1) \Delta_n}) dW_s}^{2p}}+(M\Delta_n)^{2p}\right). 
\end{equation*}
Now, using inequality (\ref{majintsigmaW2}) and since $p \geq 1,$ $
\esp{| \chi_i^n |^{2p}} \leq C_p.\left( \Delta_n^{p+1} + \Delta_n^{2p} \right) \leq C.\Delta_n^{p+1}$. Thus \eqref{E-LOC1} becomes:
\begin{equation} 
\esp{(\Lambda_1^n(t))^2} 
\leq C_p. \frac{\Delta_n^{2-p}}{h_n^2} \sum_{i\in\DNT}\left[
(\Delta_n^{p+1})^{\frac{1}{p}} \Delta_n^{p-1} + \Delta_n^{p+1}\right]\leq  \frac{C}{h_n} \left[ \Delta_n^{1+\frac{1}{p}} + \Delta_n^{2} \right], \notag %\label{E-AVOIR3}
\end{equation}
the constant $C_p$ does not depend on $t$ and as $p \geq 2$, we have,
$$
\sup_{t \in [0,T]} \, \esp{(\Lambda_1^{n}(t))^2 } \leq C_p \frac{\Delta_n^{1+\frac{1}{p}}}{h_n},
$$
which ends the proofs.
\end{proof}

\begin{proposition}  
\label{P6}
Assume $\mathbf{(SH)_2}$.
 Then, 
\begin{equation}
\label{term5}
\esp{\left|\espcond{i-1}{\left| \frac{\Delta_i^n X}{\sqrt{\Delta_n}} \right|^p}- m_p
|\sigma_{(i-1)\Delta_n}|^p \right|}
\le \begin{cases}
C\Delta_n^{\frac{1}{2}}                                  &\textit{if $p=2$}
\\  
C\Delta_n^{\frac{p-2}{2}\wedge\frac{1}{2}}   &\textnormal{if $p>2$.}
\end{cases}
\end{equation}
As a consequence, 
if $p=2$ or $p\geq 3,$ $\|\Lambda_{3}^n(t)\|_1\leq C\sqrt{\Delta_n}$ and, 
\begin{equation*}
\max\left(\sqrt{\frac{h_n}{\Delta_n}},\sqrt{\frac{1}{h_n}}\right)\Lambda_{3}^n(t) \xrightarrow[n \rightarrow +\infty]{\bL^1} 0.
\end{equation*}
\end{proposition}
\begin{proof}
We begin the proof by the following remark. Scaling and independence properties of the Brownian motion and the Ito's formula yield
\begin{align*}
m_p 
= \frac{p(p-1)}{2\Delta_n^{\frac{p}{2}}}\int_{\CN} \espcond{i-1}{|W_s-W_{(i-1)\Delta_n}|^{p-2}}ds.
\end{align*}
Keeping in mind this representation of $m_p$, we decompose the integrand of \eqref{term5} as follows:
\begin{equation*}
\espcond{i-1}{\left| \frac{\Delta_i^n X}{\sqrt{\Delta_n}} \right|^p}-  m_p
|\sigma_{(i-1)\Delta_n}|^p = A_{1,i}^n + A_{2,i}^n  \qquad \text{where}
\end{equation*}
\begin{align*}
&A_{1,i}^n =
\espcond{i-1}{\left| \frac{\Delta_i^n X}{\sqrt{\Delta_n}} \right|^p}-\frac{p(p-1)}{2\Delta_n^{\frac{p}{2}}}\int_{\CN} \espcond{i-1}{\left|\int_{(i-1)\Delta_n}^s \sigma_udW_u\right|^{p-2}\sigma^2_s}ds,\\
&A_{2,i}^n =\frac{p(p-1)}{2\Delta_n^{\frac{p}{2}}}
\int_{\CN} \espcond{i-1}{\left|\int_{(i-1)\Delta_n}^s\sigma_udW_u\right|^{p-2} \sigma^2_s -\sigma_{(i-1)\Delta_n}^p \left|\int_{(i-1)\Delta_n}^sdW_u\right|^{p-2}} ds.
\end{align*}
\noindent 

\noindent Then, the result is a consequence of Lemmas \ref{corresp51} and \ref{corresp52} corresponding to $A_{1,i}^n$ and $A_{2,i}^n$ respectively.

\begin{lemma}
\label{corresp51}
Assume  $\mathbf{(SH)_2}$.
 Then,
\begin{equation}\label{pr1}
\esp{|A_{1,i}^n|} \le \begin{cases} C.\Delta_n&\textnormal{if $p=2$,}\\
C.\Delta_n^{(\frac{p}{2}-1)\wedge\frac{1}{2}}&\textnormal{if $p>2$.}
\end{cases}
\end{equation}
\end{lemma}

\noindent \textit{Proof.}
First, we use It\^o's formula to develop $A_i^n$:
\begin{align*}
\left| \frac{\Delta_i^n X}{\sqrt{\Delta_n}} \right|^p
&=\int_{\CN}p.\text{{\rm sgn}}(X_s-X_{(i-1)\Delta_n})
\frac{|X_s-X_{(i-1)\Delta_n}|^{p-1}}{\Delta_n^{\frac{p}{2}}} a_s ds\\
&\quad+\demi p(p-1)\int_{\CN}\frac{|X_s-X_{(i-1)\Delta_n}|^{p-2}}{\Delta_n^{\frac{p}{2}}}\sigma^2_sds + M^n_i,
\end{align*}
with $\espcond{i-1}{M^n_i}=0$. It follows that:
\begin{align*}
A_{1,i}^n
&= \espcond{i-1}{\int_{\CN}p.\text{{\rm sgn}}(X_s-X_{(i-1)\Delta_n})  \frac{|X_s-X_{(i-1)\Delta_n}|^{p-1}}{\Delta_n^{\frac{p}{2}}}a_sds}\\
&\quad +\demi p(p-1) \espcond{i-1}{\int_{\CN}R^n_i(s)\sigma^2_s ds},
\end{align*}
with $R^n_i(s):=\frac{|X_s-X_{(i-1)\Delta_n}|^{p-2}}{\Delta_n^{\frac{p}{2}}}-
\frac{|\int_{(i-1)\Delta_n}^s \sigma_udW_u|^{p-2}}{\Delta_n^{\frac{p}{2}}}$. Now, using that $a$ is bounded, we have
\begin{align*}
\esp{|X_s-X_{(i-1)\Delta_n}|^{p-1}|a_s| } &\le C\left((s-(i-1)\Delta_n)^{p-1} + \esp{ \left|\int_{(i-1)\Delta_n}^s\sigma_u dW_u\right|^{p-1} }\right)\\
&\le C(s-(i-1)\Delta_n)^{p-1} + C(s-(i-1)\Delta_n)^{1\vee\frac{p-1}{2}},
\end{align*}
owing to Inequality (\ref{majintsigmaW}).  Hence, for every $p\ge2$,
$$
\esp{ \int_{\CN}  p \frac{|X_s-X_{(i-1)\Delta_n}|^{p-1}}{\Delta_n^{\frac{p}{2}}}|a_s|ds} \le
 C\Delta_n^{\frac{1}{2}\vee (2-\frac{p}{2})}.$$
\noindent Now, we observe that $R^n_i(s)=0$ when $p=2$ so the proof is ended in this case.\\

\noindent When $p>2$, recall that for every $\bar{q} > 0$ and $\forall (u,v) \in \R^2$,
\begin{equation}\label{elem1}
\VA{|u|^{\bar{q}} -|v|^{\bar{q}} } \le\begin{cases} 
|u-v|^{\bar{q}}&\textnormal{if ${\bar{q}}\le1$}\\
C_{\bar{q}}\left(|u-v||u|^{{\bar{q}}-1}+|u-v|^{{\bar{q}}}\right)&\textnormal{if ${\bar{q}}>1$,}
\end{cases}
\end{equation}
applying it with $\bar{q}=p-2$ yields 
\begin{equation}\label{ineq34}
|R^n_i(s)| \le 
\begin{cases} \frac{1}{\Delta_n^{\frac{p}{2}}}|\int_{(i-1)\Delta_n}^s a_u du|^{p-2} &\textnormal{if $p\le3$}\\
C.\frac{1}{\Delta_n^{\frac{p}{2}}} \left( \VA{\int_{(i-1)\Delta_n}^sa_udu}.\VA{\int_{(i-1)\Delta_n}^s\sigma_udW_u}^{p-3} + \VA{\int_{(i-1)\Delta_n}^s a_udu}^{p-2}\right) &\textnormal{if $p>3$.}
\end{cases}
\end{equation}
First, let $p\in(2,3]$. Since $a$ is uniformly bounded,
$$
|R_i^n(s)| \le C\Delta_n^{-\frac{p}{2}}\left[(s-(i-1)\Delta_n)^{p-2~}\right].
$$ 
Then, $\esp{\sigma_s^3}$  uniformly bounded, Cauchy-Schwarz and inequality (\ref{majintsigmaW}) yield
\begin{equation} \label{lala1}
\esp{ \int_{\CN}|R^n_i(s)\sigma^2_s|ds }\le C\Delta_n^{\frac{p}{2}-1}.
\end{equation}
Assume now that $p>3$. First, for all $s \in[(i-1)\Delta_n,i\Delta_n],$ we derive from $a$
bounded and  Cauchy-Schwarz inequality  that 
\begin{align*}
&\esp{\VA{\int_{(i-1)\Delta_n}^s a_udu} \VA{\int_{(i-1)\Delta_n}^s\sigma_udW_u}^{p-3}\sigma_s^2}\leq C(s-(i-1)\Delta_n)\esp{\VA{\int_{(i-1)\Delta_n}^s\sigma_udW_u}^{2(p-3)}}^{\frac{1}{2}} 
\esp{\VA{\sigma_s^{4}}}^{\frac{1}{2}}.
\end{align*}
Therefore, using inequalities \eqref{majintsigmaW} and \eqref{EspSup}, we have:
$$
\esp{ \VA{\int_{(i-1)\Delta_n}^sa_udu}\sigma_s^2 \left| \int_{(i-1)\Delta_n}^s\sigma_udW_u \right|^{p-3} }\leq C.(s-(i-1)\Delta_n)^{\frac{p-3}{2}+1}.
$$
Thus, we derive from \eqref{ineq34}, the preceding inequality and \eqref{lala1} that when $p>3$,
\begin{align*}
\esp{ \int_{\CN}|R^n_i(s)\sigma^2_s|ds} 
&\leq \frac{C}{\Delta_n^{\frac{p}{2}}}
\int_{\CN} \left[ (s-(i-1)\Delta_n)^{p-2} + (s-(i-1)\Delta_n)^{\frac{p-1}{2}} \right] ds\leq C.\Delta_n^{\demi}.
\end{align*}
\end{proof}
\noindent We now focus on $ A_{2,i}^n$.
\begin{lemma} \label{corresp52}
Assume  $\mathbf{(SH)_2}$.
Then,
\begin{equation} \label{pr2}
\esp{| A_{2,i}^n|} \le\begin{cases} C.\Delta_n^{\frac{1}{2}}&\textnormal{if $p=2$,}\\
 C.\Delta_n^{(\frac{p}{2}-1)\wedge\frac{1}{2}}&\textnormal{if $p>2$.}
 \end{cases}
\end{equation}
\end{lemma}

\begin{proof}
In case $p=2$ we deal with $\frac{1}{\Delta_n}\int_{\CN}\left(\sigma_s^2-\sigma_{(i-1)\Delta_n}^2\right)ds$. Hence by Cauchy-Schwarz, (\ref{EspSup}) and (\ref{difsigmap}), we deduce that,
\begin{align}
\esp{\frac{1}{\Delta_n}  \left|\int_{\CN}\left(\sigma_s^2-\sigma_{(i-1)\Delta_n}^2 \right) ds \right|} 
&\leq C \frac{1}{\Delta_n} \int_{\CN} \left[\esp{ |\sigma_s-\sigma_{(i-1)\Delta_n}|^{2}} \right]^{\demi} \le  C \Delta_n^\demi.\label{pr3}
\end{align}
When $p>2,$ first,
\begin{align*}
&\frac{1}{\Delta_n^{\frac{p}{2}}}\int_{\CN}  \espcond{i-1}{\sigma^2_s \VA{\int_{(i-1)\Delta_n}^s \sigma_udW_u}^{p-2}-\sigma_{(i-1)\Delta_n}^p \VA{\int_{(i-1)\Delta_n}^s dW_u}^{p-2}} ds \nonumber= B_{1,i}^n + B_{2,i}^n \qquad\text{with},\\
&B_{1,i}^n =\frac{1}{\Delta_n^{\frac{p}{2}}}\int_{\CN}\espcond{i-1}{(\sigma^2_s-\sigma^2_{(i-1)\Delta_n})\left|\int_{(i-1)\Delta_n}^s \sigma_udW_u\right|^{p-2}}ds,\\
&B_{2,i}^n =\frac{\sigma^2_{(i-1)\Delta_n}}{\Delta_n^{\frac{p}{2}}} \int_{\CN} \espcond{i-1}{\VA{\int_{(i-1)\Delta_n}^s \sigma_udW_u}^{p-2} -
\VA{\int_{(i-1)\Delta_n}^s \sigma_{(i-1)\Delta_n}dW_u}^{p-2}} ds.
\end{align*}
Let us focus on $B_{1,i}^n$ and let $\bar{q}>1$ and $\bar{r}>1$ satisfying $\frac{1}{\bar{q}}+\frac{1}{\bar{r}}=1$ and ${\bar{r}}> 2\vee \frac{2}{p-2}$. Using H\"older inequality, we have
\begin{align*}
&\esp{\VA{\sigma^2_s-\sigma^2_{(i-1)\Delta_n}}.\VA{\int_{(i-1)\Delta_n}^s \sigma_udW_u}^{p-2}}
\leq \left(\esp{\VA{\sigma^2_s-\sigma^2_{(i-1)\Delta_n}}^{\bar{q}}}\right)^{\frac{1}{{\bar{q}}}} . \left(\esp{ \VA{\int_{(i-1)\Delta_n}^s \sigma_udW_u}^{{\bar{r}}(p-2)}}\right)^{\frac{1}{{\bar{r}}}}.
\end{align*}
Then, on the one hand, applying again Holder's inequality applied with $\tilde{p}=2/\bar{q}$($>1$) and $\tilde{q}=\bar{q}/(\bar{q}-2)$, we derive from  (\ref{EspSup}) and (\ref{difsigmap}),
$$
\esp{|\sigma^2_s-\sigma^2_{(i-1)\Delta_n}|^{\bar{q}}}^{\frac{1}{{2\bar{q}}}}
\leq C \esp{|\sigma_s-\sigma_{(i-1)\Delta_n}|^{2}}^{\frac{1}{2}}
\leq C (s-(i-1)\Delta_n)^{\frac{1}{2}},
$$
On the other hand, using (\ref{majintsigmaW}),
$$\left(\esp{ \VA{\int_{(i-1)\Delta_n}^s \sigma_udW_u}^{{\bar{r}}(p-2)}}\right)^{\frac{1}{{\bar{r}}}}\leq
C(s-(i-1)\Delta_n)^{\frac{p-2}{2}}.$$
Thus,
\begin{equation}
\esp{ \VA{\sigma^2_s-\sigma^2_{(i-1)\Delta_n}}.\VA{\int_{(i-1)\Delta_n}^s \sigma_udW_u}^{p-2}}
\leq
C.(s-(i-1)\Delta_n)^{\frac{p-2}{2}+\frac{1}{2}} .
\end{equation}
Hence, we have
\begin{equation}
\label{st3}
\esp{\VA{B_{1,i}^n}}  \leq C\Delta_n^\frac{1}{2}.
\end{equation}
We now study $B_{2,i}^n$. Set $M_s^n=\int_{(i-1)\Delta_n}^s(\sigma_u-\sigma_{(i-1)\Delta_n})dW_u$. By \eqref{elem1},
\begin{align}
&\VA{ \VA{\int_{(i-1)\Delta_n}^s\sigma_udW_u}^{p-2}- \VA{\int_{(i-1)\Delta_n}^s\sigma_{(i-1)\Delta_n}dW_u}^{p-2}}
 \nonumber\\
&\qquad \le
\begin{cases}|M_s^n|^{p-2} &\textnormal{if $p\le 3$}\\
C.\left(|M_s^n|.\VA{\int_{(i-1)\Delta_n}^s\sigma_{(i-1)\Delta_n}dW_u}^{p-3} +|M_s^n|^{p-2}\right) &\textnormal{if $p> 3$}.
\end{cases}\label{part17}
\end{align}

\noindent Hence, if $p\le3$, it follows from \eqref{majintsigmaW2} and Cauchy-Schwarz inequality that
\begin{align} 
\esp{ \VA{B_{2,i}^n} }
&\leq \frac{C}{\Delta_n^{\frac{p}{2}}}\int_{(i-1)\Delta_n}^{i\Delta_n} \esp{|M_s^n|^{2(p-2)}}^\frac{1}{2} .\esp{|\sigma_{(i-1)\Delta_n}|^{4}}^{\frac{1}{2}} ds, \notag\\
&\le \frac{C}{\Delta_n^{\frac{p}{2}}}\int_{(i-1)\Delta_n}^{i\Delta_n} [(s-(i-1)\Delta_n)^{2(p-2)}]^{\frac{1}{2}}ds\le C\Delta_n^{\frac{p}{2}-1}.\label{st2}
\end{align}
Assume now that $p>3$. According to \eqref{part17}, we have two terms to manage with. On the one hand, by Cauchy-Schwarz and \eqref{majintsigmaW2}, we have
\begin{align*}
&\esp{ \sigma_ {(i-1)\Delta_n}^2|M_s^n|. \VA{\int_{(i-1)\Delta_n}^s  \sigma_ {(i-1)\Delta_n}dW_u}^{p-3}}\\ 
&\qquad\le \left( \esp{|M_s^n|^2} \right)^\demi  (s-(i-1)\Delta_n)^{\frac{p-3}{2}} \esp{ |\sigma_{(i-1)\Delta_n}|^{2(p-1)}}^{\frac{1}{2}} \le C(s-(i-1)\Delta_n)^{\frac{p-1}{2}}.
\end{align*}
On the other hand, Cauchy-Schwarz and \eqref{majintsigmaW2} applied with $q=2(p-2)\ge2$ yield
\begin{align*}
\ES\left[\sigma_ {(i-1)\Delta_n}^2|M_s^n|^{p-2}\right] 
&\le \left(\esp{\sigma_ {(i-1)\Delta_n}^4}\right)^{\frac{1}{2}}.\left(\esp{|M_s^n|^{2(p-2)}}\right)^{\frac{1}{2}}\le C(s-(i-1)\Delta_n)^\frac{p-1}{2}.
\end{align*}
Thus, it follows that when $p>3$, 
\begin{equation}\label{st1}
\esp{|B_{2,i}^n|} \le C.\Delta_n^{\frac{1}{2}}.
\end{equation}
Finally, we derive the lemma from \eqref{st3}, \eqref{st2}, \eqref{st1}.
\end{proof}

\subsection{Proof of main Theorems, a synthesis}
Gathering the previous steps, Theorems \ref{thprincipal} and \ref{thprincipal2} are now consequences of the classical following lemma:

\begin{lemma} \label{lemstbl}
Let $(X_n)$ and $(Y_n)$ be some sequences of random variables defined on $(\Omega,{\cal F},\PE)$ with values in a Polish space $E$. Assume that $(X_n)$ converges ${\cal L}-s$ to $X$ and that $(Y_n)$ converges in probability to $Y$. Then, the sequence of random variables $(Z_n=X_n+Y_n)$ converges ${\cal L} - s$ to $X+Y$.
\end{lemma}

\noindent Indeed, focus for instance on statements \eqref{premiereconvergence} and \eqref{tioru}. By Lemma \ref{appoxM}, it is enough to prove these convergences under $\mathbf{(SH)_q}$ and $\mathbf{(H_q^2)}$.  Then, on the one hand, using Proposition \ref{lambda5} (and the fact that 
$\sup_{n\ge1}(\sqrt{h_n/\Delta_n})h_n^{1/q}<+\infty$ under the assumptions), Lemma \ref{L4} and Proposition \ref{P6} with $p\in\{1/2\}\cup[3,+\infty[$, we deduce respectively that
\begin{align*}
&\sqrt{\frac{h_n}{\Delta_n}}\Lambda_5^{n}(t)\xrightarrow[n\rightarrow+\infty]{\PRB} 0, \quad \sqrt{\frac{h_n}{\Delta_n}}\Lambda_1^{n}(t) \xrightarrow[n\rightarrow+\infty]{\PRB} 0,\\
&\sqrt{\frac{h_n}{\Delta_n}}\Lambda_3^{n}(t)\xrightarrow[n\rightarrow+\infty]{\PRB} 0\quad\textnormal{since $\Delta_n^{\frac{p-2}{2}\wedge \frac{1}{2}}=\Delta_n^\frac{1}{2}$ if $p\ge3$}.
\end{align*}
On the other hand, under the assumptions of Theorems \ref{thprincipal}$(i)$ and \ref{thprincipal2}$(i)$, one deduces from Proposition \ref{L5}$(i)$ and $(ii)$ respectively that,
\begin{equation}
\sqrt{\frac{h_n}{\Delta_n}}\left(\Lambda_2^{n}(t)+\Lambda_4^n(t)\right) \xrightarrow[n \rightarrow +\infty]{{\cal L}-s} f(t,p)U,
\end{equation}
Therefore, \eqref{premiereconvergence} and \eqref{tioru} follow from Lemma \ref{lemstbl} applied  with $Y=0$ and from 
the decomposition of the error stated in Section \ref{sectiondecomp}. Using Proposition \ref{L5}$(i)$ when $\sqrt{\Delta_n}/{h_n}\rightarrow\beta\in\R_+$, the same ideas lead to \eqref{convergebis}. Finally, applying Lemma \ref{lemstbl}
with $Y=\theta_.^0$, \eqref{partconv}
follows from Proposition \ref{L5}$(ii)$ and from the fact that
$$\frac{1}{h_n}\Lambda_5^{n}(t)\xrightarrow[n\rightarrow+\infty]{\PRB} \theta_t^0\quad
 \textnormal{(see Proposition \ref{lambda5})}.$$

\section{Asymptotic confidence interval}
\label{sec4}
Actually, Theorems \ref{thprincipal} and \ref{thprincipal2} allow us to build a confidence region to estimate for all $t$ parameter $\sigma_t.$
Since the variance limits in their second part depend on the unknown parameters
$\eta_1(t)$ and $\eta_2(t),$ we focus on their first part (i) when $h_n/\sqrt{\Delta_n}$ (respectively $\limsup_{n\rightarrow+\infty}h_n^{1/2+1/q}/\sqrt{\Delta_n}<+\infty$).
This confidence region could be defined as follows:
$$
\left\{
\sigma_t,~~\frac{\sqrt{r_n}|\Sigma(p,\Delta_n,h_n)_t-\sigma_t^p|}{\sqrt{\varphi_1(p,t,\sigma_t)}}\leq 1.96
\right\},
$$
and according to (\ref{premiereconvergence}) or (\ref{tioru}) with, for instance, asymptotic probability $0.95,$ we get:
$$
\pr{}{\frac{\sqrt{r_n}|\Sigma(p,\Delta_n,h_n)_t-\sigma_t^p|m_p}{\sigma_t^p\sqrt{m_{2p}-m_p^2}}\leq 1.96 }
\xrightarrow[n\rightarrow +\infty]{} 0.95.
$$
Thus, with  $0.95$ asymptotic confidence, 
\begin{equation}
\label{CI}
\sigma_t^p \in \left[ \frac{m_p\sqrt{r_n} \, \Sigma(p,\Delta_n,h_n)_t}{m_p\sqrt{r_n} +1.96\sqrt{m_{2p}-m_p^2}} ~,~ \frac{m_p\sqrt{r_n} \,\Sigma(p,\Delta_n,h_n)_t}{m_p\sqrt{r_n} -1.96\sqrt{m_{2p}-m_p^2}} \right].
\end{equation}
The confidence  interval length is about $r_n^{-\demi}.$ Actually, the most interesting point is that we obtain an asymptotic  confidence  interval for the relative error:
$$
\pr{}{\left|\frac{\Sigma(p,\Delta_n,h_n)_t}{\sigma_t^p}-1\right|\leq \frac{1.96\sqrt{m_{2p}-m_p^2}}{m_p\sqrt{r_n}}}
\xrightarrow[n\rightarrow \infty]{} 0.95.
$$
\begin{remark}
\label{inflence/p}
Finally, to compare this result with respect to $p$, we have to compare asymptotic confidence
intervals of $\sigma_t,$ depending on $p$, namely
$$\sigma_t \in \left[ \left(\frac{m_p\sqrt{r_n} \, \Sigma(p,\Delta_n,h_n)_t}{m_p\sqrt{r_n} +1.96\sqrt{m_{2p}-m_p^2}}\right)^{\frac{1}{p}} ~,~ \left(\frac{m_p\sqrt{r_n} \,\Sigma(p,\Delta_n,h_n)_t}{m_p\sqrt{r_n} -1.96\sqrt{m_{2p}-m_p^2}}\right)^{\frac{1}{p}} \right],
$$
This interval length is about $r_n^{-\frac{1}{2}}\frac{\sqrt{m_{2p}-m_p^2}}{pm_p},$ and this length order is 
unhappily increasing with  $p,$ so it could be not so good to use $p>2.$
\end{remark}
%
%Je ne vois pas pourquoi ce résultat serait attendu.... Il n'est pas vrai
%que la variance croit avec $p$ ! voir par exemple une var plus petite que 1 en module...
%De plus, nous n'avons pas une variance de $|U|^p$ mais son écart-type divisé par $p$
%fois son espérance.... alors ??
%\\
%\\
%J'ai coupé la section test et l'ai copiée sous end-document.}
%
\section{Simulations}
\label{sec6}
In this section, we want to test numerically the volatility estimator.
In order to be able to compare the estimations with the true volatility, we do not use some real datas but get our observations from \textit{quasi-exact} simulations of toy models
(by quasi-exact, we mean simulations of the process using an Euler scheme with a very small time-discretization step).

\subsection{A numerical test in a continuous stochastic volatility model}
In this part, we consider the stochastic volatility model proposed in \cite{fouque} where the volatility is an Ornstein-Uhlenbeck process. Denote the price  by $(S_t)$ and by $(\sigma_t)$ the (non-negative) stochastic volatility. Set $X_t:=\log(S_t)$ and $v_t:=\sigma^2_t$. The model is defined by:
\begin{equation*}
\begin{cases}
dX_t=(r-\frac{1}{2}\sigma^2_t)dt+\sigma(t)dW_t^1\\
dv_t=a(m-v_t)dt+\beta(\rho dW^1_t+\sqrt{1-\rho^2}dW_t^2),
\end{cases}
\end{equation*}
where $r,a,\beta$ and $m$ are some positive parameters, $\rho\in[-1,1]$ and the processes $W^1$ and $W^2$ are independent one-dimensional Brownian motions.\\
We set $X_0=\log(50)$, $v_0=m$ and simulate \textit{quasi-exactly} $(X_t,v_t)$ at times $0,1/n,2/n,\ldots,1$ with the following parameters:
\begin{equation*}\label{paramfunc}
 r=0.05,\quad\rho=0,\quad a=1,\quad m=0.05,\quad \textnormal{and}\quad \beta=0.05.
 \end{equation*}
Using the simulated observations $X_0,X_{1/n},\ldots,X_1$, we compute the estimator $\Sigma(p,1/n,h_n)$ on $[h_n,1]$ and compare its value with the true volatility. In Figures \ref{figure1} and \ref{figure2}, we represent the corresponding graphics  for $n=1000$ and $n=10000$ and $h_n=n^{-1/2}$. In all the figures, we choose $p=2$ 
since as shown in the computation of the confidence interval length
in Remark \ref{inflence/p}, to increase $p$ is not a good choice. The process $(\sigma_t)$ is plotted as continuous line whereas the estimator $\Sigma(2,1/n,h_n)$ is plotted as discontinuous line.

\begin{center}
\begin{figure}[htbp]\centering
\includegraphics[width=12cm,height=4cm]{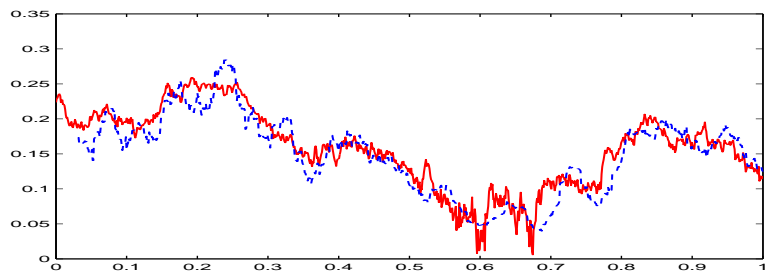}
\caption{$n=1000$, $h_n=n^{-1/2}$.}
\label{figure1}
\end{figure}

\begin{figure}[htbp]\centering
\includegraphics[width=12cm,height=4cm]{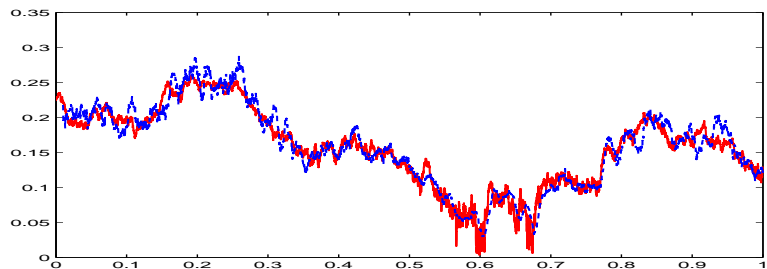}
\caption{$n=10000$, $h_n=n^{-1/2}$.}
\label{figure2}
\end{figure}
\end{center}

By  Remark \ref{rk15}, taking  $r_n=n^{\rho}$ with $\rho\in(0,1/2)$ and $p\in\{2\}\cap(5/2,+\infty)$ (or equivalently $h_n=n^{\rho-1}$), we obtain a rate of order $n^{\rho/2}$. In particular, we can derive that the best rate is obtained in the limit case $\rho=1/2$. This theoretical result is confirmed in the following computation. 
Denote by $E_n(p,h_n)$ the mean relative error defined by: 
\begin{equation}\label{erreurrelative}
E_n(p,h_n):=\frac{1}{n}\sum_{k=1}^n\frac{\left|\Sigma(p,n^{-1},h_n)_{k/n}^{1/p}-\sigma(\frac{k}{n})\right|}{\sigma(\frac{k}{n})}.
\end{equation}
We obtain the following results:
\begin{center}
\begin{tabular}{|*{4}{c|}}
   \hline
    &$E_n(2,n^{-0.4})$ &$E_n(2,n^{-0.5})$&$E_n(2,n^{-0.6})$ \\
    \hline
$n=10^3$ & 18,9\%  &$\bf 16,6\%$ &  18,6\% \\
\hline
$n=10^4$ & 12,2\%    &$\bf 11,0\%$&  12,3\%\\
 \hline 
 \hline
 &$E_n(4,n^{-0.4})$&$E_n(4,n^{-0.5})$&$E_n(4,n^{-0.6})$\\
 \hline
$n=10^3$ &20,3\%&$\bf 17,5\%$& 19,2\%\\
\hline
$n=10^4$ & 13,0\%&$\bf 11,9\%$&12,9\%\\
 \hline 
\end{tabular}
\end{center}
Here, Remark \ref{inflence/p} is confirmed by the fact: the estimations seem to be better with $p=2$ than with $p=4$.

\subsection{A numerical test in a jump model}
In this last part, we assume that the volatility is a jump process solution to a SDE driven by a tempered stable subordinator 
$(Z_t^{(\lambda,\beta)})$ with L\'evy measure $\pi(dy)=\1_{y>0}\exp(-\lambda y)/y^{1+\beta}dy$. This model can be viewed as a particular case of the Barndorff-Nielsen and Shephard model \cite{barndorff} (for other jump volatility models see $e.g.$ \cite{ConTan,vives06}):
\begin{equation*}
\begin{cases}
dX_t=(r-\frac{1}{2}\sigma^2(t))dt+\sigma(t)dW_t^1\\
dv_t=-\mu v_tdt+dZ_t^{(\lambda,\beta)}
\end{cases}
\end{equation*}
with the following choice of parameters:
$$r=0.05,\quad\mu=1,\quad\lambda=1 \quad\textnormal{and}\quad\beta=1/2.$$

This concerns Example (\ref{exlevy}) in Section \ref{sec1}
where Hypotheses $\mathbf{(H^1_q)}$ and $\mathbf{(H^2_q)}$ hold for any $q\geq\beta$ and Theorem \ref{thprincipal2} can
be applied.
As in the preceding example, we simulate  $(X_t,v_t)$ on the interval $[0,1]$ with $X_0=\log(50)$ and $v_0=0.05$. In order to compare the two types of models,  we chose some similar parameters. The main difference between these two models comes from the variations which are stronger in the first case. 
We obtain  a quasi-exact sequel $(X_{k/n},v_{k/n})$ with  $k\in\{0,\ldots,n\}$. In Figures \ref{figure5} and \ref{figure6}, we represent the estimated and true volatilities for some different choices of $h_n=n^{-1/2}$, $n=10^3$ and $n=10^4$.

\begin{center}
\begin{figure}[htbp]\centering
\includegraphics[width=12cm,height=4cm]{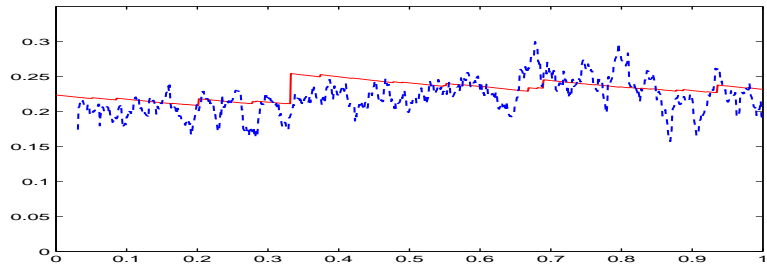}
\caption{$n=1000$, $h_n=n^{-1/2}$.}
\label{figure5}
\end{figure}

\begin{figure}[htbp]\centering
\includegraphics[width=12cm,height=4cm]{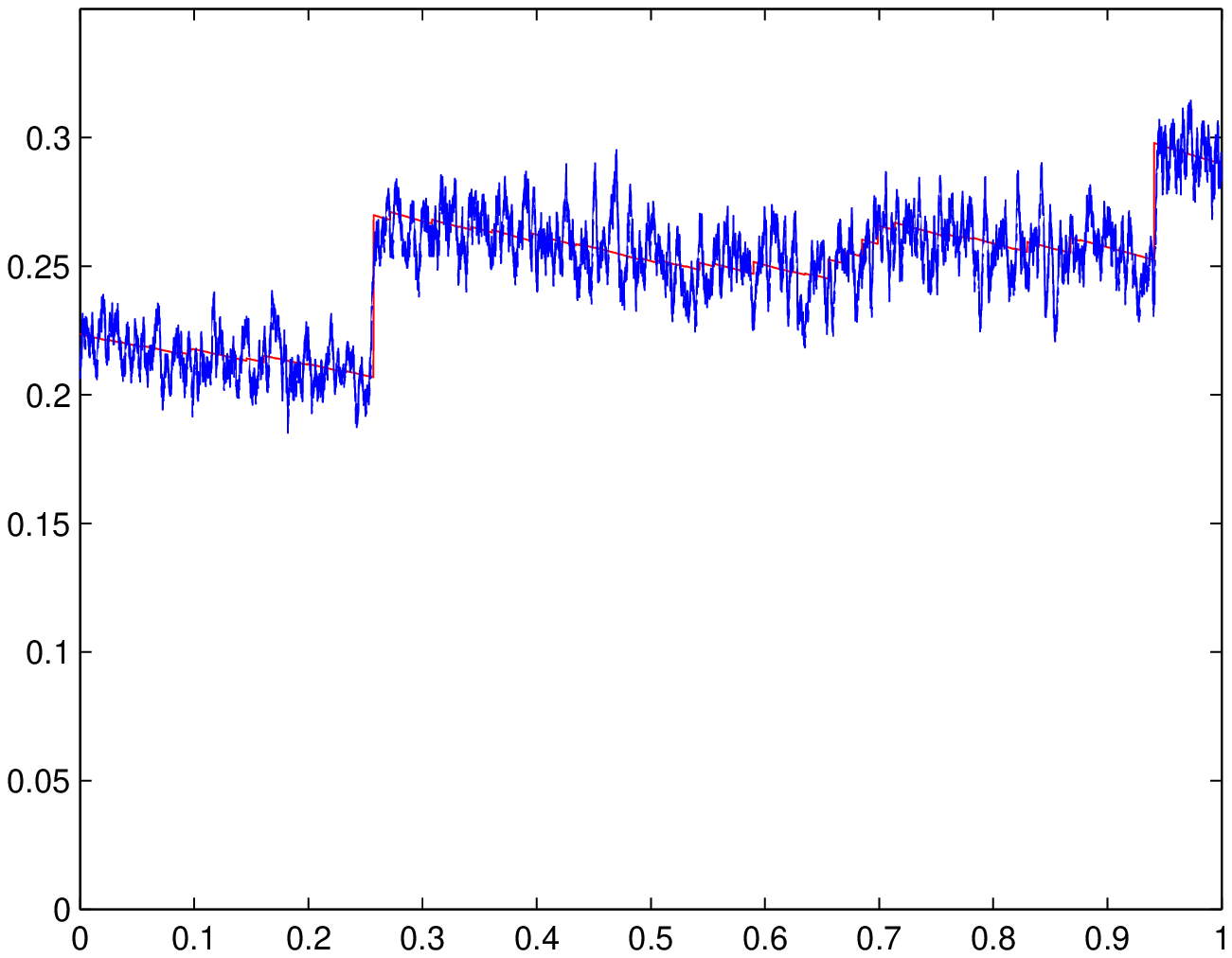}
\caption{$n=10000$, $h_n=n^{-1/2}$.}
\label{figure6}
\end{figure}
\end{center}

For these computations, we obtain the following mean relative errors:
\begin{center}
\begin{tabular}{|*{4}{c|}}
   \hline
    & $E_n(2,n^{-0.4})$ &$E_n(2,n^{-0.5})$&$E_n(2,n^{-0.6})$ \\
    \hline
$n=10^3$ & 13,2\%  &8,3\% &$\bf  6,3\%$ \\
\hline
$n=10^4$ & 9,1\%    & 5,5\%&$\bf  3,2\%$\\ 
\hline 
\hline
         &$E_n(4,n^{-0.4})$&$E_n(4,n^{-0.5})$&$E_n(4,n^{-0.6})$ \\
\hline
$n=10^3$ &15,5\%&11,0\%& $\bf 8,8\%$\\
\hline
$n=10^4$ & 10,1\%&6,6\%&$\bf 3,9\%$\\ 
\hline
\end{tabular}
\end{center}
It seems the best result is obtained with $h_n=n^{-0.6}$, according to Remark \ref{rk15} in case $\eta_1=\eta_2=0$: the best convergence rate is obtained
with $\rho=2/3.$
\section{Appendix}\label{preuveannexe}
\noindent \textbf{Proof of Lemma \ref{appoxM}:}
Since the arguments are almost the same for each statement of the main results, we only  prove \eqref{premiereconvergence} (with $p\in\{2\}\cup[3,+\infty)$ and $t\ge0$.)\\
Let $(X,\sigma)$ satisfy $\mathbf{(H_q^1)}$ and $\mathbf{(H_q^2)}$. Then, there exists a sequence  $(T_M)_{M\ge1}$ of stopping times increasing to $\infty$ such that the processes $(a_t)$, $(b_t)$, $(\eta_1(t))$, $(\eta_2(t))$ and 
$(\int (1\wedge |y|^q) F_t(dy))$ are bounded on $[0,T_M]$. Since $(\int_{|y|\ge1} y F_s(dy))_{s\ge0}$ is locally bounded, we can also assume that $|\Delta Y_t|\le M$ on $[0,T_M]$.
Then, let $X^M$, $\sigma^M$ be defined by $X^M=X_{t\wedge T^M}$ and $\sigma^M=|Y^M|$ where
\begin{align*}
&Y^M_t=\int_0^{t\wedge T_M} b_s ds+\int_0^{t\wedge T_M} \eta_1(s) dW^1_s+\int_0^{t\wedge T_M} \eta_2(s) dW^2_s\\
&+\int_0^{t\wedge T_M}\int_{\{|y|\le1\}} y(\mu-\nu)(ds,dy)+\int_0^{t\wedge T_M}\int_{\{1\le |y|\le M\}}y\mu(ds,dy).
\end{align*}
By construction, $(X^M,Y^M)$ satisfies $\mathbf{(SH)_q}$ and $\mathbf{(H^2_q)}$ for every $M\in\N$.
It follows from the assumptions of the proposition that for every $M\in\N$,
\begin{equation}\label{convergencestoppee}
\sqrt{r_n}\left(\Sigma^M(p,\Delta_n,h_n)_t-(\sigma^M_t)^p\right)) \xrightarrow[n \rightarrow +\infty]{{\cal L}-s} \sqrt{\varphi_1(p,t,\sigma)} U,
\end{equation}
where $\Sigma^M$ is the statistic related to 
$(X^M,\sigma^M)$ as in (\ref{defSigma}),
 $U\sim {\cal N}(0,1)$ and $U$ is independent of ${\cal F}_{t}$. Let us now prove \eqref{premiereconvergence}.
Let $g$ be  a  bounded continuous function on $\R$ and let $H$ be a bounded ${\cal F}$-measurable random variable. Then, for all $M\in\N$,
\begin{align*}
&\esp{H g(\sqrt{r_n}(\Sigma(p,\Delta_n,h_n)_t-(\sigma_t)^p)}-\tilde{\mathbb{E}}[g(\sqrt{\varphi_1(p,t,\sigma)} U)] =\\&
\esp{H g\left(\sqrt{r_n}(\Sigma(p,\Delta_n,h_n)_t-\sigma_t^p)\right)}-
\esp{ H g\left(\sqrt{r_n}(\Sigma^M(p,\Delta_n,h_n)_t-(\sigma^M_t)^p)\right)}\\
&+ \esp{ H g\left(\sqrt{r_n}(\Sigma^M(p,\Delta_n,h_n)_t-(\sigma^M_t)^p)\right) }- \tilde{\mathbb{E}}[ H g\left( \sqrt{\varphi_1(p,t,\sigma^M)} U\right)]\\
&+ \tilde{\mathbb{E}}[H g\left( \sqrt{\varphi_1(p,t,\sigma^M)} U\right)]- \tilde{\mathbb{E}}[H g\left( \sigma_t,\sqrt{\varphi_1(p,t,\sigma)} U\right)].
\end{align*}
Set $B_M=\{\omega, t+h_1<T_M(\omega)\}.$ By construction,  on $B_M$,
 $\sigma^M_t=\sigma_t$ and $\Sigma^M(p,\Delta_n,h_n)_t=\Sigma(p,\Delta_n,h_n)_t$
 for every $n\in\N$. Thus, uniformly in $n$, the  first and third right-hand side terms are bounded by $2\|g\|_\infty\|H\|_\infty \pr{}{B_M^c}$ for every $M\in\N.$
Then, since Assumption $\mathbf{(H_q^1)}$ implies that $T_M\rightarrow+\infty$ $a.s$, $\pr{}{B_M^c} \rightarrow 0$ as $M\rightarrow+\infty$. Now, by \eqref{convergencestoppee}, for every $M\in\N$,
$$
\esp{ H g\left(\sqrt{r_n}(\Sigma^M(p,\Delta_n,h_n)_t-(\sigma^M_t)^p)\right)} \xrightarrow[n\rightarrow+\infty]{} \tilde{\mathbb{E}}[ H g\left( \sqrt{\varphi_1(p,t,\sigma^M)} U\right)]$$
and the result follows.\\

\noindent \textbf{Proof of Lemma \ref{lemmecontrol}}
(i). Let us prove \eqref{EspSup}. Thanks to Jensen's inequality, we can only consider the case $r\ge2$. Since the jumps of $Y$ are bounded, we can compensate the big jumps and write
$$
Y_t=\int_0^t\tilde{b}_udu+\int_0^t\eta_1(s) dW_s+\int_0^t\eta_2(s) dW^2_s+\int_0^t
\int_{\R} y (\mu-\nu)(ds,dy),
$$
where $\tilde{b}_t=b_t+\int_{\{|y|>1\}} y F_t(dy).$ Then,
using $\mathbf{(SH)_2}$ and Burkholder-Davis-Gundy inequality, we have for every $r\ge2$:
$$
\esp{\sup_{t\in[0,T]}|\sigma_t|^r}\le C\left(T^r+T^{r/2}+\esp{\sup_{t\in[0,T]}\left(\int_0^t\int_{\R}
y^2\mu(ds,dy)\right)^{r/2}}\right),
$$
where $C$ is a deterministic constant. Let us focus on the last term of the right-hand side. We can write:
\begin{align*}
&\esp{\sup_{t\in[0,T]}\left(\int_0^t\int_{\R}
y^2\mu(ds,dy)\right)^{\frac{r}{2}}}\\&\qquad \le C\esp{\sup_{t\in[0,T]}\left(\int_0^t\int_{\R}
y^2(\mu-\nu)(ds,dy)\right)^{\frac{r}{2}}} +C\esp{\left(\int_0^T\int_{\R} y^2\nu(ds,dy)\right)^{\frac{r}{2}}}\\
&\qquad\le C \left(\esp{\left(\int_0^T\int_{\R}
y^4\mu(ds,dy)\right)^{\frac{r}{4}}}+T^{\frac{r}{2}}\right),
\end{align*}
where in the last inequality, we again used Burkholder-Davis-Gundy inequality and $\mathbf{(SH)_2}$.
Set $k_0=\min\{k\in\N, r\le 2^k\}$. By an iteration, we obtain
\begin{align*}
&\esp{\sup_{t\in[0,T]}\left(\int_0^t\int_{\R} y^2\mu(ds,dy)\right)^{r/2}}\le C \left(\esp{\left(\int_0^T\int_{\R}
y^{2^{k_0}}\mu(ds,dy)\right)^{r/2^{k_0}}}
+C\sum_{i=1}^{k_0}T^{r/2^{k}}\right).
\end{align*}
Using that $|u+v|^\rho\le |u|^\rho+|v|^\rho$ when $\rho\le 1$, we have
$$
\esp{\sup_{t\in[0,T]}\left(\int_0^t\int_{\R} y^{2^{k_0}}\mu(ds,dy)\right)^{r/2^{k_0}}}\le
\esp{\left(\int_0^T\int_{\R} y^{r}\mu(ds,dy)\right)}\le C_T,
$$
and the result follows. \\
\noindent(ii). For \eqref{difsigmap}, when $r \ge 2$, we obtain by a similar approach:
$$
\esp{|\sigma_t-\sigma_s|^r~|~{\cal F}_s} \le C_T (|t-s|^r+|t-s|^{r/2})+|t-s|),
$$
where $C_T$ is a deterministic constant. This yields the result when $r \ge 2$. When $r<2$ the result follows from the Jensen inequality. Let us  prove \eqref{majintsigmaW}. If $0<q<2$, using Jensen inequality and the concavity of the map $x\mapsto x^{q/2}$, we have
$$
\esp{ \left|\int_s^t\sigma_udW_u\right| ^q}\le \esp{\left|\int_s^t\sigma_u dW_u\right|^{2}}^{\frac{q}{2}} \leq  \left(\int_s^t\esp{\sup_{s\le u\le t}\sigma_u^2}du \right)^{q/2}\le C(t-s)^{\frac{q}{2}},
$$
owing to \eqref{EspSup}. When $q\ge2$, we first derive from Burkholder-Davis-Gundy inequality that
$$
\esp{ \left|\int_s^t\sigma_udW_u\right| ^q} \leq C \esp{\left|\int_s^t\sigma^2_udu\right| ^{q/2}}. 
$$
Then,  Jensen inequality and the convexity of the map $x\mapsto x^{q/2}$ yield
$$
\left( \frac{1}{t-s} \int_s^t \sigma_u^2du \right)^{q/2} \leq \frac{1}{t-s} \int_s^t \sigma_u^q du,
$$
Thus,
$$
\esp{ \left |\int_s^t\sigma_u^2du \right|^{q/2}}\leq (t-s)^{q/2-1}\int_s^t \esp{ \sigma_u^q}du,
$$
and \eqref{majintsigmaW} again follows from \eqref{EspSup}.\\
\noindent Finally, let us prove \eqref{majintsigmaW2}. With similar arguments as previously,
we obtain:
\begin{equation*}
\esp{ \left|\int_s^t(\sigma_u-\sigma_s)dW_u\right| ^q} \leq C\begin{cases}
 \left(\int_s^t\esp {|\sigma_u-\sigma_s|^2}du\right)^{q/2} &\textnormal{if $q\in(0,2]$}\\
 (t-s)^{q/2-1}\int_s^t\esp{|\sigma_u-\sigma_s|^q}du&\textnormal{if $q\ge 2$,}
 \end{cases}
 \end{equation*}
and \eqref{majintsigmaW2} follows from \eqref{difsigmap}. \\

\noindent \textbf{Acknowledgements.} We are deeply grateful  to  the listeners  of our presentations
and to Jean Jacod for their valuable advices.

\end{document}